\documentclass[runningheads,svgnames,svgnames]{llncs}
\usepackage[T1]{fontenc}
\usepackage{verbatim}
\usepackage{graphicx, tikz, subcaption, float}
\usepackage{framed, tcolorbox, xcolor}
\usepackage{parskip, nicefrac}
\usepackage{thmtools, amsmath, amssymb}
\usepackage{hyperref, cleveref}
\usepackage{lmodern}

\graphicspath{{pictures/}}

\usepackage{xcolor}
\definecolor{DarkGreen}{rgb}{0.1,0.5,0.1}
\hypersetup{
colorlinks=true,
linkcolor=IndianRed,
urlcolor=DarkGreen,
citecolor=SeaGreen
}

\newcommand{\density}[1]{den(#1)}
\newcommand{\dominatingSet}{\ensuremath{\gamma(G)}}
\newcommand{\romanDominationNumber}{\ensuremath{\gamma_R(G)}}
\newcommand{\italianDominationNumber}{\ensuremath{\gamma_I(G)}}
\newcommand{\lp}{\reflectbox{\rotatebox[origin=c]{0}{$\vdash$}}}
\newcommand{\rp}{\reflectbox{\rotatebox[origin=c]{180}{$\vdash$}}}

\begin{document}
\title{$m$-Eternal Domination and Variants\\ on Some Classes of Finite and Infinite Graphs}
\titlerunning{$m$-Eternal Domination on Finite and Infinite Graphs}
%
\author{T. Calamoneri$^1$\orcidID{0000-0002-4099-1836}, F. Cor\`o$^2$\orcidID{0000-0002-7321-3467}, \\
N. Misra$^3$\orcidID{0000-0003-1727-5388}, S.G. Nanoti$^3$\orcidID{0009-0009-7789-8895}, \\
G. Paesani$^1$\orcidID{0000-0002-2383-1339}}
%
\authorrunning{T. Calamoneri et al.}
%
\institute{Sapienza University of Rome, Italy \and
University of Padua, Italy \and
Indian Institute of Technology, Gandhinagar, India
}
%
\maketitle              
\begin{abstract}
We study the {\sc $m$-Eternal Domination} problem, which is the following two-player game between a defender and an attacker on a graph: initially, the defender positions $k$ guards on vertices of the graph; the game then proceeds in turns between the defender and the attacker, with the attacker 
selecting
a vertex and the defender 
responding to
the attack by moving a guard to the attacked vertex. The defender may move more than one guard on their turn, but guards can only move to neighboring vertices. 
The defender wins a game on a graph $G$ with $k$ guards if the defender has a strategy 
such that at every point of the game the vertices occupied by guards form a dominating set of $G$
and 
the attacker wins otherwise.
The {\em $m$-eternal domination number} of a graph $G$ is the smallest value of $k$ for which $(G,k)$ is a defender win. 

We show that {\sc $m$-Eternal Domination} is NP-hard, as well as some of its variants, even on special classes of graphs.
We also show structural results for the {\sc Domination} and {\sc $m$-Eternal Domination} problems in the context of four types of 
infinite regular grids: square, octagonal, hexagonal, and triangular, establishing tight bounds.

\end{abstract}
%
%
\section{Introduction}

A subset $S$ of vertices in a simple undirected graph $G$ is called a \emph{dominating set} if every vertex outside of $S$ has a neighbor in $S$. 
Finding a smallest-sized dominating set is a fundamental computational problem, and indeed, several variations of this basic premise have been considered in the literature~\cite{HHS98}.

In this work, our focus is on the $m$-eternal version of the {\sc Domination} problem, where we think of the vertices of the dominating set as being occupied by ``guards'' that can move in response to ``attacks'' on the vertices. Specifically, we consider the following two-player graph game. To begin with, the defender places $k$ guards on vertices of the graph. The game continues with alternating turns between an attacker and the defender. On each turn, the attacker chooses a vertex to attack, and the defender responds by repositioning guards, ensuring at least one guard moves onto the attacked vertex. The defender may relocate multiple guards during a turn, but each guard can only move to an adjacent vertex. The defender wins if there exists a strategy that maintains a dominating set of vertices occupied by guards throughout the entire game. Otherwise, the attacker wins.The $m$-eternal domination number of a finite graph $G$ is the smallest value of $k$ for which the defender wins.

Two 
variations of the 
domination
problem are {\sc Roman Domination} and {\sc Italian Domination}. A Roman dominating function on a graph \( G = (V, E) \) is a function \( f: V \rightarrow \{0, 1, 2\} \) such that every vertex \( v \) with \( f(v) = 0 \) has at least one neighbor \( u \) with \( f(u) = 2 \). The weight of a Roman dominating function is the sum \( \sum_{v \in V} f(v) \). The Roman domination number \( \gamma_R(G) \) is the minimum weight of a Roman dominating function on \( G \). This concept is inspired by the defensive strategy of the Ancient Roman Empire, where guards were stationed such that every unguarded position was adjacent to a doubly guarded position. An Italian dominating function on a graph \( G = (V, E) \) is a function \( f: V \rightarrow \{0, 1, 2\} \) such that for every vertex \( v \) with \( f(v) = 0 \), the sum of \( f(u) \) over all neighbors \( u \) of \( v \) is at least 2. The weight of an Italian dominating function is the sum \( \sum_{v \in V} f(v) \). The Italian domination number \( \gamma_I(G) \) is the minimum weight of an Italian dominating function on \( G \). This concept generalizes Roman domination by allowing more flexible assignments of resources to vertices and their neighbors.

Given the context of warfare, it is very natural to propose studying the ``eternal'' variations of Roman and Italian domination. Presumably, regions will be attacked more than once, and the guards will have to find ways of reconfiguring themselves so that they maintain the defense invariants that they started with. 

\paragraph{Our Contributions.} In \Cref{ss:cc}, we consider the computational complexity of the {\sc $m$-Eternal Domination} problem and the Roman and Italian variants in the $m$-eternal setting. We show that {\sc $m$-Eternal Domination} (and its connected variant, where we require the locations of the guards to induce a connected subgraph) is NP-hard even on bipartite graphs of diameter four (\Cref{thm:eternalDominationNPHard}); and the {\sc $m$-Eternal Roman Domination} and {\sc $m$-Eternal Italian Domination} problems are NP-hard even on split graphs (Theorems \ref{thm:eternalRomanDominationNPHard} and \ref{thm:eternalItalianDominationNPHard}). 

It is well known~\cite{BDEMY17} that for split graphs, the eternal domination number is at most its domination number plus one, and a characterization of the split graphs which achieve equality is given. The authors also show that the decision versions of {\sc Domination} and {\sc $m$-Eternal Domination} are NP-complete, even on Hamiltonian split graphs. Moreover, computing the eternal domination number can be solved in polynomial time for any subclass of split graphs for which the domination number can be computed in polynomial time, in particular for strongly chordal split graphs.

In \Cref{ss:inf}, we focus on structural results for the {\sc Domination} and {\sc $m$-Eternal Domination} problems in the context of four types of infinite regular grids: square, octagonal, hexagonal, and triangular. We prove $m$-eternal dominating sets that are optimal according to a parameter (that will be defined later) expressing the concept of minimality in infinite graphs.

We highlight that exact results for the eternal domination number are given in \cite{FMV15,GKM13,DM17} and, more recently, in \cite{FMV20}, for grids with either 2 or 3 rows, and in \cite{VV16} for grids with 4, 5, or 6 rows. When $G$ is a general $n \times m$ square grid, it is clear that the eternal domination number must be at least as large as the domination number, so by the result in \cite{GPRT11}, it must be at least $\lfloor \frac{(n-2)(m-2)}{5} \rfloor -4$;
the best-known upper bound was determined in \cite{LMS19}, and is $\frac{mn}{5} +O(n+m)$, thus showing that the eternal domination number is within $O(m+n)$ of the domination number.

\section{Preliminaries}

In this section, we introduce definitions and terminology that will be relevant to our discussions later. We will be dealing with simple undirected graphs denoted by \( G = (V, E) \). The \emph{degree} of a vertex \( v \)
is the number of edges incident to~\( v \). An {\em independent set} or {\em clique} is a subset of vertices such that no or every possible edge is present, respectively. The \emph{neighborhood} of a vertex \( v \), denoted \( N(v) \), is the set of all vertices adjacent to \( v \). The \emph{closed neighborhood} of a vertex $v$, denoted $N[v]$, is the set $N(v) \cup \{v\}$. A graph is \emph{bipartite} if its vertex set \( V \) can be partitioned into two disjoint {\em independent sets} \( U \) and \( W \). A \emph{split graph} is a graph in which the vertex set \( V \) can be partitioned into a clique and an independent set. In this paper, we say that a graph $G=(V,E)$ is infinite if $V$ is countably infinite. For more detailed graph terminology, we refer to \cite{D12}, while for computational complexity terminology, we refer to \cite{GJ79}.

\subsection{Domination, $m$-eternal dominations and variants}

We first introduce the definitions of domination, Roman domination, and Italian domination. Then, we give the concept of an $m$-eternal domination game and define all the previous kinds of domination under this light.

\begin{definition}
\cite{B62,O62}
A \emph{dominating set} for a graph \(G = (V, E)\) is a subset \(D \subseteq V\) such that every vertex not in \(D\) is adjacent to at least one vertex in \(D\). The \emph{domination number} \(\dominatingSet\) is the minimum cardinality of a dominating set in \(G\). If the subgraph induced by $D$ in $G$ is connected, then $D$ is said to be a {\em connected dominating set}.
\end{definition}

The {\sc Domination} problem (
determine a minimum dominating set for graphs) is NP-hard even for planar graphs \cite{GJ79} and bipartite graphs of diameter three \cite{MS24}.

\begin{definition} 
\cite{R97}
A \emph{Roman dominating function} on a graph \(G = (V, E)\) is a function \(f: V \to \{0, 1, 2\}\) such that every vertex \(v\) for which \(f(v) = 0\) is adjacent to at least one vertex \(u\) for which \(f(u) = 2\). The \emph{Roman domination number} \(\romanDominationNumber\) is the minimum weight \(\sum_{v \in V} f(v)\) of a Roman dominating function on \(G\).
If the subgraph induced by $D=\{ v\,\, | \,f(v) \neq 0\}$ in $G$ is connected, then $D$ is said to be a {\em connected Roman dominating set}.
\end{definition}

The {\sc Roman Domination} problem (which is to determine the Roman domination number of a graph) is NP-hard even for split graphs, bipartite graphs, and planar graphs \cite{Cal04}, for chordal graphs \cite{LC13}, and for subgraphs of grids \cite{NS15}.

\begin{definition}
\cite{Cal16,HK17}
An \emph{Italian dominating function} on a graph \(G = (V, E)\) is a function \(f: V \to \{0, 1, 2\}\) such that, for every vertex \(v\) for which \(f(v) = 0\), the sum of the function values of the neighbors of \(v\) is at least $2$. The \emph{Italian domination number} \(\italianDominationNumber\) is the minimum weight \(\sum_{v \in V} f(v)\) of an Italian dominating function on \(G\).
If the subgraph induced by $D=\{ v\,\, | \,f(v) \neq 0\}$ in $G$ is connected, then $D$ is said to be a {\em connected Italian dominating set}.
\end{definition}

The {\sc Italian Domination} problem (which is to determine the Italian domination number of a graph) is NP-hard even for bipartite planar graphs, chordal bipartite graphs \cite{FL23}, and split graphs \cite{chen2018note}.

The {\em $m$-eternal domination game} is a two-player turn-based game on graph $G$: a team of guards initially occupies a dominating set on a graph $G$. An \emph{attacker} then assails a vertex without a guard on it; the \emph{defender} exploits the guards to defend against the attack: one of the guards has to move to the attacked vertex from one of its neighbors, 
while the defender can choose whether to move or not the remaining guards to one of its neighbor vertices. This attack-defend procedure continues eternally. The defender wins if they can eternally maintain a dominating set against any sequence of attacks; otherwise, the attacker wins.

\begin{definition} \cite{GHH05}
An {\em $m$-eternal domination set} for a graph $G=(V, E)$ is a subset $D \subseteq V$ where a defender can place its guard and win the $m$-eternal domination game.
The \emph{$m$-eternal domination number} of $G$, denoted by $\gamma^\infty(G)$, is the minimum number of guards required to defend an indefinite sequence of attacks. 
\end{definition}

We use the term \emph{configuration} to refer to a state of the $m$-eternal domination game or guard positions. Note that when configurations correspond to dominating functions rather than sets, the value of the function at a vertex indicates the number of guards occupying the vertex, and again multiple guards are allowed to move in one round; no vertex accommodates more than two guards in a configuration, and guards can only move to neighbor vertices. 
By requiring the configuration of guards at every step to correspond to a Roman dominating function or to an Italian dominating function instead of a dominating set, we obtain the \emph{$m$-eternal Roman domination} and \emph{$m$-eternal Italian domination games}, and the minimum numbers of guards necessary to defend for ever $G$ are the \emph{$m$-eternal Roman domination} and \emph{$m$-eternal Italian domination numbers} and denoted by $\gamma^\infty_R(G)$ and $\gamma^\infty_I(G)$, respectively. 
Similarly to what we did for $m$-domination, Roman and Italian domination, it is possible to define $m$-eternal connected domination and $m$-eternal Roman and Italian connected domination.

In this paper, we do not consider variations of these games where only one guard is allowed to move in one step (as, {\em e.g.}, in \cite{Bal04}). To distinguish from this variation, the $m-$ prefix highlights that multiple guards are allowed to move. 

%
%
We focus on the \textsc{$m$-Eternal Domination} computational problem, consisting in taking as input a graph $G$ and determining $\gamma^\infty(G)$.
The \textsc{$m$-Eternal Roman Domination} and \textsc{$m$-Eternal Italian Domination} problems are defined analogously.



We make here also a couple of remarks on terminology. When we refer to the {\sc $m$-Eternal Domination} problem (and variants), unless mentioned otherwise, we are referring to the optimization version of the question, as opposed to explicitly calling it the minimization problem. Also, we use the term \emph{configuration} to refer to a state of the $m$-eternal domination game or guard positions. 

\subsection{Infinite Regular Grids}

We work with four kinds of infinite regular grids, shown in~\Cref{fig:grids}. 
In the following definitions and the rest of the paper, we adopt the following convention. We imagine that the vertices of the infinite grids occupy integer coordinates of the Cartesian plane.   
Despite this, we do not mean to consider the grids as graphs embedded in a metric space. We make this assumption to describe the moving strategies in an easier way.

We formally define the infinite regular grids as follows: 
\begin{itemize}
\item {\bf infinite square grid $T_4$} (see~\Cref{fig:infinite-grid-square}): every vertex $(x,y)$ is connected with $(x+1, y)$, $(x-1, y)$, $(x, y+1)$ and $(x, y-1)$;
\item {\bf infinite octagonal grid $T_8$} (see~\Cref{fig:infinite-grid-octagonal}): every vertex $(x,y)$ is connected with $(x+1, y)$, $(x-1, y)$, $(x, y+1)$, $(x, y-1)$, $(x-1, y-1)$, $(x+1, y+1)$, $(x-1, y+1)$ and $(x+1, y-1)$;
\item {\bf infinite hexagonal grid $T_3$} (see~\Cref{fig:infinite-grid-hex}): every vertex $(x,y)$ is connected both with $(x, y+1)$ and with $(x, y-1)$; moreover, every vertex $(x,y)$ is connected to $(x+1,y)$ if $x+y$ is even and to $(x-1,y)$ otherwise;
\item {\bf infinite triangular grid $T_6$} (see~\Cref{fig:infinite-grid-tri}): every vertex $(x,y)$ is connected with $(x+1, y)$, $(x-1, y)$, $(x, y+1)$, $(x, y-1)$, $(x-1, y-1)$ and $(x+1, y+1)$.
\end{itemize}

\begin{figure}[ht]
\centering
\subfloat[$T_4$\label{fig:infinite-grid-square}]{
\includegraphics[scale=0.65]{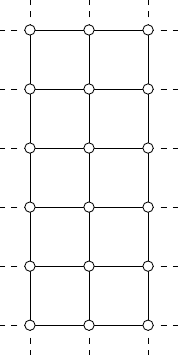}
}~
\subfloat[$T_8$\label{fig:infinite-grid-octagonal}]{
\includegraphics[scale=0.65]{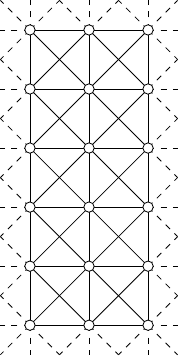}
}
\subfloat[$T_3$\label{fig:infinite-grid-hex}]{
\includegraphics[scale=0.65]{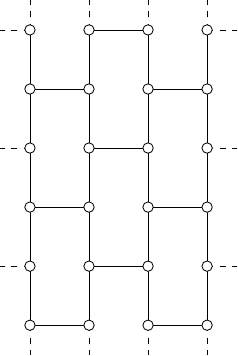}
}~
\subfloat[$T_6$\label{fig:infinite-grid-tri}]{
\includegraphics[scale=0.65]{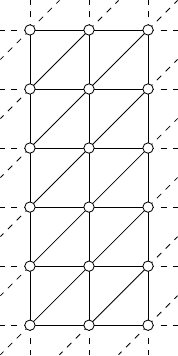}
}~
\caption{The four considered infinite regular grid graphs.}
\label{fig:grids}
\end{figure}

\subsection{Optimal Domination}


For infinite graphs, the meaning of a minimum set w.r.t. a set of properties is far from clear. Some related work uses notions of {\em density} that can be applied to infinite graphs w.r.t. different properties, {\em e.g.}, for cops configurations in Cops and Robber games~\cite{BHW07} or for eternal vertex cover configurations on regular grids~\cite{CC22}. In the context of dominating sets for infinite graphs, we choose to avoid working with a density-based approach and instead adopt a local cardinality-based approach of determining when a proposed dominating set is optimal. Given two dominating sets $S$ and $T$, we say that $S$ is \emph{at least as good as} $T$ if there is an injective map from $S$ to $T$, that is, the cardinality of $S$ is smaller or equal to that of $T$. A dominating set is {\em optimal} if it is at least as good as every other dominating set of $G$. We now turn to a stronger notion of optimality.

Consider an infinite graph $G$ with finite maximum degree. Given a subset $S$ of $V(G)$ and a vertex $v \in V(G)$, we define the \emph{domination index} of $v$ w.r.t. $S$ as $|S \cap N[v]|$. A dominating set $S$ is \emph{strongly optimal} if, for every vertex $v$ of the graph, the domination index of $v$ w.r.t. $S$ is exactly one\footnote{In the context of finite graphs, a dominating set such that every vertex is uniquely dominated is called an exact dominating set, and because of connections with coding theory, in some settings such dominating sets are called perfect codes.}. Note that a strongly optimal dominating set may not exist: for example, consider the infinite complete bipartite graph $G$ with $V(G)= \{a_i~|~i \in\mathbb{N}\} \cup \{b_i~|~i \in \mathbb{N}\}$ and $E(G)= \{(a_i,b_j)~|~ i,j \in \mathbb{N}\}$. However, we note that a strongly optimal dominating set, if it exists, is also optimal:

\begin{lemma}
If $S$ is a strongly optimal dominating set and $T$ is any dominating set of an infinite graph $G$ with finite maximum degree, then $S$ is at least as good as $T$.
\end{lemma}

\begin{proof}
Consider the following map $f: S \rightarrow T$ defined as follows: for every vertex $v \in S$, let $f(v)$ be an arbitrarily chosen vertex from $N[v] \cap T$. 
\\
We first show that $f$ is well-defined. Since $T$ is a dominating set, for every $v \in S$, we have $N[v] \cap T \neq \emptyset$, so there always exists a vertex to which $v$ can be mapped. Next, we prove that $f$ is injective. Suppose, by contradiction, that there exist two distinct vertices $u, v \in S$ such that $f(u) = f(v) = w$ for some $w \in T$. This implies that $w \in N[u] \cap N[v]$, and consequently, the domination index of $w$ w.r.t. $S$ would be at least two. However, this contradicts the fact that $S$ is a strongly optimal dominating set, where every vertex must have a domination index of exactly one.
Therefore, $f$ is an injective map from $S$ to $T$, proving that $S$ is at least as good as $T$.\qed
\end{proof}

To recap, we have the following definitions:

\begin{definition}
    Let $G$ be an infinite graph with finite maximum degree and $S \subseteq V(G)$. The \emph{domination index} of a vertex $v \in V(G)$ w.r.t. $S$ is defined as $|S \cap N[v]|$. A dominating set $S$ is \emph{optimal} if for any other dominating set $T$ of $G$, there exists an injective map from $S$ to $T$. A dominating set $S$ is \emph{strongly optimal} if for every vertex $v \in V(G)$, the domination index of $v$ w.r.t. $S$ is exactly one.
\end{definition}

\section{Complexity Results for Variants of the $m$-Eternal Domination}\label{ss:cc}

In this section, we prove that the considered problems are all NP-hard, even for interesting subclasses of graphs. 

We preliminarily give a property whose analogous for the $m$-eternal vertex cover problem has been proved in \cite{KM09}.

\begin{lemma}\label{lem:connectedcore}
Let \(G\) be a graph and \(Z\) be a connected dominating set of \(G\). Then $\gamma^\infty(G)\leq |Z|+1$.
\end{lemma}

\begin{proof}
Initially, place a guard on each vertex of $Z\cup \{v\}$, where $v$ is any vertex not belonging to $Z$. We prove that this configuration is an $m$-eternal dominating set. At every step of the attacking sequence, we will maintain the following invariant: there are guards on $Z$ and one extra guard anywhere else, which we will refer to as the {\it floating guard}, positioned on vertex $v$. Whenever any vertex $v'$ that does not have a guard is attacked, let $v_0,\ldots v_k$ be any path such that $v_0=v$, $v_k=v'$ and $v_i\in Z$ for every $i= 1, \ldots, k-1$, for some $k\geq 1$. For every $i<k$, the guard on $v_i$ moves to $v_{i+1}$. After that, each vertex of $Z\cup \{v'\}$ is occupied by a guard, thus restoring both the invariant and successfully defending the attack. The size of $Z\cup \{v\}$ is $|Z|+1$ by construction, and since $G$ is finite, the size is preserved through sequences of dominating sets defending the attack sequence.\qed
\end{proof}

\begin{restatable}{theorem}{eternalDominationNPHard}
\label{thm:eternalDominationNPHard}
The \textsc{$m$-Eternal Domination} and \textsc{$m$-Eternal Connected Domination} problems are NP-hard, even on bipartite graphs of diameter $4$.   
\end{restatable}

\begin{proof}
We describe here a reduction from the \textsc{Dominating Set} problem. Let \(\langle G = (V,E); k\rangle\) be an instance of \textsc{Dominating Set}, and let  \(V=\{v_1, \ldots, v_n\}\). Construct a bipartite graph \(H=(A \cup B, F)\) as follows: let $U=\{u_1, \ldots, u_n\}$ and $W=\{w_1, \ldots, w_n\}$ be two copies of $V$; let \(\{p_1, \ldots,\) \(p_{n+1}\}=P\) and $w$ be new vertices. We define \(A=U\cup P\) and \(B=W\cup \{w\}\). The set of edges $F$ is defined as follows: \(u_i\) is adjacent to \(w_j\) iff \((v_i,v_j) \in E\) or \(i=j\); moreover, $(w, u_i)$ and $(w, p_j)$ are both in $F$ for each $i=1, \ldots, n$ and $j=1, \ldots , n+1$.
An example of this construction is shown in~\Cref{fig:EDS}. Now we consider the instance \(\langle H;k+2\rangle\) for the \textsc{$m$-Eternal Domination} and argue the equivalence of these instances.

 \paragraph*{Forward Direction.} 
 Let \(S \subseteq V\) be a dominating set of \(G\) of size at most \(k\). Initially, place a guard on each \(u_i\) for every \(v_i \in S\). Denote this set of vertices which are occupied by guards in $A$ as $T$.  
 The set $T\cup \{w\}$ induces a connected subgraph on $H$ because every vertex in $A$ is adjacent to $w$. This set also forms a dominating set of $H$ because every vertex in $A$ is adjacent to $w$, and every vertex in $B$ has a neighbor in $T$ because $S$ is a dominating set of $G$. By~\Cref{lem:connectedcore}, $\gamma^\infty(H)\leq |T\cup \{w\}|+1=k+2$.

\paragraph*{Reverse Direction.} Suppose $H$ has a winning strategy with $k+2$ guards, and consider one of its configurations such that $P$ and the vertex $w$ have exactly one guard each: such configuration can be obtained from any other configuration after a vertex of $P$ that does not have a guard is attacked. 
Let $S$ be the subset of $V$ such that $v_i\in S$ if and only if $u_i$ or $w_i$ has a guard. Clearly, $S$ has size at most $k$ and is a dominating set for $G$. 
Indeed, suppose there exists a vertex $v_i\in V$ that has no neighbor in $S$; thus, 
$w_i$ has no guard, and has no neighbor that has a guard and that the considered configuration is not part of a winning strategy, 
a contradiction.\\
 Finally, we consider the instance \(\langle H;k+2\rangle\) for the \textsc{$m$-Eternal Connected Domination}. To prove the equivalence of these instances, we reuse the previous proof and note that all 
 configurations 
 in a winning strategy are connected.\qed 
 \end{proof}

\begin{figure}[ht]
\centering
\begin{minipage}{0.25\textwidth}
\begin{tikzpicture}[scale=0.65]
\coordinate (A) at (-1,1); \coordinate (B) at (-1,-1); \coordinate (C) at (1,-1); \coordinate (D) at (1,1); \coordinate (E) at (0,2); 
\draw[color=black](A)--(B)--(C)--(D)--(E)--(A)--(D);
\draw[fill=black] (B) circle [radius=3pt] (D) circle [radius=3pt];
\draw[fill=white] (A) circle [radius=3pt] (C) circle [radius=3pt] (E) circle [radius=3pt];
\node[left] at (A) {$v_1$}; \node[left] at (B) {$v_2$}; \node[right] at (C) {$v_3$}; \node[right] at (D) {$v_4$}; \node[above] at (E) {$v_5$};
\end{tikzpicture}
\end{minipage}
\begin{minipage}{0.5\textwidth}
\begin{tikzpicture}[scale=0.7]
\coordinate (A1) at (-2,1); \coordinate (A2) at (-1,1); \coordinate (A3) at (0,1); \coordinate (A4) at (1,1); \coordinate (A5) at (2,1); \coordinate (B1) at (-2,-1); \coordinate (B2) at (-1,-1); \coordinate (B3) at (0,-1); \coordinate (B4) at (1,-1); \coordinate (B5) at (2,-1); \coordinate (Z) at (4,-1); \coordinate (P1) at (3,1); \coordinate (P2) at (3.5,1); \coordinate (P3) at (4,1); \coordinate (P4) at (4.5,1); \coordinate (P5) at (5,1); \coordinate (P6) at (5.5,1);
\draw (P1)--(Z)--(P2)(P3)--(Z)--(P4)(P6)--(Z)--(P5)(A1)--(Z)--(A2)(A3)--(Z)--(A4)(A5)--(Z)(B2)--(A1)--(B1)--(A2)(B4)--(A1) (B1)--(A4)(B3)--(A2)--(B2)--(A3)(B4)--(A3)--(B3)--(A4)(B5)--(A3) (B3)--(A5)--(B5)(B5)--(A4)--(B4)--(A5) (A1)--(B5) (A5)--(B1)
(-2.5,0.7) rectangle (6,1.5)(-2.5,-1.5) rectangle (4.5,-0.7);
\draw[fill=black](A2) circle [radius=3pt](A4) circle [radius=3pt](Z) circle [radius=3pt](P1) circle [radius=2pt];
\draw[fill=white] (A1) circle [radius=3pt](A3) circle [radius=3pt](A5) circle [radius=3pt](B1) circle [radius=3pt](B2) circle [radius=3pt](B3) circle [radius=3pt](B4) circle [radius=3pt](B5) circle [radius=3pt](P2) circle [radius=2pt](P3) circle [radius=2pt](P4) circle [radius=2pt](P5) circle [radius=2pt](P6) circle [radius=2pt];
\node[above] at (A1) {$u_1$};\node[above] at (A2) {$u_2$};\node[above] at (A3) {$u_3$};\node[above] at (A4) {$u_4$};\node[above] at (A5) {$u_5$};\node[above] at (P1) {$p_1$};\node[above] at (P2) {$p_2$};\node[above] at (P3) {$p_3$};\node[above] at (P4) {$p_4$};\node[above] at (P5) {$p_5$};\node[above] at (P6) {$p_6$};\node[below] at (B1) {$w_1$};\node[below] at (B2) {$w_2$};\node[below] at (B3) {$w_3$};\node[below] at (B4) {$w_4$};\node[below] at (B5) {$w_5$};\node[below] at (Z) {$w$};\node[right] at (6,1.1) {$A$};\node[right] at (4.5,-1.1) {$B$};
\end{tikzpicture}
\end{minipage}
\caption{The construction for the proof of~\Cref{thm:eternalDominationNPHard}. On the left, graph $G$ for which the black vertices represent a (minimum size) dominating set. On the right, auxiliary bipartite graph $H=(A, B; F)$ for which the black vertices represent a configuration of an $m$-eternal dominating set.}
\label{fig:EDS}  
\end{figure}
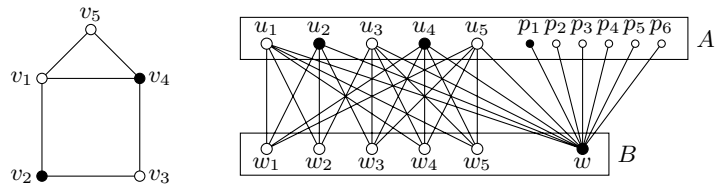

\begin{restatable}{lemma}{connectedRcore}
\label{lem:connectedRcore}%
Let \(G\) be a graph and \(Z\) be a connected dominating set of \(G\). Then $\gamma^\infty_R(G)\leq 2|Z|+1$.
\end{restatable}

\begin{proof}
    Let $f(\cdot)$ be the function defined as follows: $f(u)=2$ where $u\in Z$ and $f(u)=0$ where $u\notin Z$. The function $f(\cdot)$ is a Roman dominating function for $G$ because every vertex with value $0$ is not in $Z$, but it has a neighbor in $Z$ with value $2$ since $Z$ is a dominating set. 
    
    The set of configurations of an $m$-eternal Roman dominating set is defined by functions $f_v$ (also inducing Roman dominating sets) as follows: for each $v \notin Z$, the function $f_v$ is defined as $f_v(u)=f(u)$ if $u\neq v$ and $f_v(v)=1$.
    %
    At every step of the attacking sequence, we will maintain the following invariant: there are two guards on each vertex in $Z$ and one extra guard anywhere, which we will refer to as the {\it floating guard}, positioned on vertex $v$. Whenever any vertex $v'$ which does not have a guard is attacked, let $v_0,\ldots v_k$ be any path such that $v_0=v$, $v_k=v'$ and $v_i\in Z$ for every $i= 1, \ldots, k-1$, for some $k\geq 1$. For every $i<k$, a guard on $v_i$ moves to $v_{i+1}$. After that, each vertex of $Z$ is occupied by two guards, and the vertex $v'$ is occupied by one guard, thus both restoring the invariant and successfully defending the attack. 
    
    The weight of $f_v$ is equal to $2|Z|+1$ by construction, and since $G$ is finite, the weight is preserved through the sequences of Roman dominating functions defending the sequence of attacks.
    \qed    
\end{proof}

\begin{restatable}{theorem}{eternalRomanDominationNPHard}
    \label{thm:eternalRomanDominationNPHard}
    The \textsc{$m$-Eternal Roman Domination} and \textsc{$m$-Eternal Connected Roman Domination} problems are NP-hard, even on split graphs.
\end{restatable}

\begin{proof}
We describe here a reduction from the \textsc{Dominating Set} problem. Let \(\langle G = (V,E); k\rangle\) be an instance of \textsc{Dominating Set} and let \(V=\{v_1, \ldots, v_n\} \). Construct a split graph \(H=(A \cup B, F)\) as follows: let $A=\{u_1, \ldots, u_n\}$ be a copy of $V$ and let \(B = \{w_i^{(j)}\}\) for \(1 \leqslant i \leqslant n\) and \(1 \leqslant j \leqslant 2n+2\) contain \(2n+2\) copies of \(V\). The edge set $F$ is defined as follows: \(A\) induces a clique and \(B\) induces an independent set; moreover, \(u_i\) is adjacent to \(w_k^{(j)}\) for all \(j\) iff \((v_i,v_k) \in E\) and \(u_i\) is adjacent to \(w_i^{(j)}\) for all \(i,j\).
An example of this construction is shown in~\Cref{fig:RDS}. Now we consider the instance \(\langle H; 2k+1\rangle\) for the \textsc{$m$-Eternal Roman Domination} and argue on the equivalence of these instances.    

\paragraph{Forward Direction.} Let \(S \subseteq V\) be a dominating set of size at most \(k\). Let $T$ be the set of vertices in $A$ that corresponds to $S$, {\em i.e}. $u_i\in T$ if and only if $v_i\in S$. The set $T$ induces a connected subgraph because it is a subset of $A$, which induces a clique in $H$. By construction, $T$ is also a dominating set of $H$, because $S$ is a dominating set of $G$.
By~\Cref{lem:connectedRcore}, $\gamma^\infty_R(H)\leq 2|T|+1\leq 2k+1$.

\paragraph{Reverse Direction.} Suppose that $H$ has a winning strategy with at most \(2k+1\) guards, and consider an initial configuration induced by a Roman dominating function $f$ of $H$ such that the weight is equal to $2k+1$. For every \(1 \leqslant j \leqslant 2n+2\), let \(W^{(j)} := \{w_i^{(j)}~|~ 1 \leqslant i \leqslant n\}.\) Moreover, let $S\subseteq A$ be the set of vertices in \(A\) that have value two each, {\em i.e.}, $s \in S$ if and only if $f(s)=2$. By the pigeon-hole principle, there exists a \(1 \leqslant j \leqslant 2n+2\) such that \(W^{(j)}\) contains only zero value vertices, {\em i.e.}, \(f(w_i^{(j)})=0\) for every \(i \leqslant n\). Since \(W^{(j)}\) only has neighbors in \(A\), set \(S\) must be a dominating set of \(G\), and this means that \(G\) has a dominating set \(S\) of size at most \(k\).\\
Finally, we consider the instance \(\langle H;2k+1\rangle\) for the \textsc{$m$-Eternal Connected Roman Domination}. To prove the equivalence of these instances, we reuse the previous proof and note that all considered configurations involved in a winning strategy are connected.\qed
\end{proof}

\begin{figure}[ht]
\centering
\begin{tikzpicture}[scale=0.6]
\coordinate (A1) at (-2,1); \coordinate (A2) at (-1,1); \coordinate (A3) at (0,1); \coordinate (A4) at (1,1); \coordinate (A5) at (2,1); \coordinate (B1) at (-2,-1); \coordinate (B2) at (-1,-1); \coordinate (B3) at (0,-1); \coordinate (B4) at (1,-1); \coordinate (B5) at (2,-1); \coordinate (Z) at (4,-1); \coordinate (C1) at (3.5,-1); \coordinate (C2) at (4.5,-1); \coordinate (C3) at (5.5,-1); \coordinate (C4) at (6.5,-1); \coordinate (C5) at (7.5,-1); 
\coordinate (D1) at (10,-1); \coordinate (D2) at (11,-1); \coordinate (D3) at (12,-1); \coordinate (D4) at (13,-1); \coordinate (D5) at (14,-1); 
\draw[dashed] (2.3,0.7)--(4,-0.7)(2.5,0.9)--(10.4,-0.7);
\draw (A1)--(B1)--(A2)(B4)--(A1) (B1)--(A4)(B3)--(A2)--(B2)--(A3)(B4)--(A3)--(B3)--(A4)(B5)--(A3) (B3)--(A5)--(B5)(B5)--(A4)--(B4)--(A5) (A1)--(B5) (A5)--(B1)
(-2.5,0.7) rectangle (2.5,1.7)(-2.5,-1.9) rectangle (2.5,-0.7)(3,-1.9) rectangle (8,-0.7)(9.4,-1.9) rectangle (14.8,-0.7);
\draw[fill=black](A2) circle [radius=3pt](A4) circle [radius=3pt];
\draw[fill=gray] (B1) circle [radius=3pt];
\draw[fill=white] (A1) circle [radius=3pt](A3) circle [radius=3pt](A5) circle [radius=3pt](B2) circle [radius=3pt](B3) circle [radius=3pt](B4) circle [radius=3pt](B5) circle [radius=3pt] (C1) circle [radius=3pt](C2) circle [radius=3pt](C3) circle [radius=3pt](C4) circle [radius=3pt](C5) circle [radius=3pt] circle [radius=3pt] (D1) circle [radius=3pt](D2) circle [radius=3pt](D3) circle [radius=3pt](D4) circle [radius=3pt](D5) circle [radius=3pt];
\node[above] at (A1) {$u_1$};\node[above] at (A2) {$u_2$};\node[above] at (A3) {$u_3$};\node[above] at (A4) {$u_4$};\node[above] at (A5) {$u_5$};\node[below] at (B1) {$w^{(1)}_1$};\node[below] at (B2) {$w^{(1)}_2$};\node[below] at (B3) {$w^{(1)}_3$};\node[below] at (B4) {$w^{(1)}_4$};\node[below] at (B5) {$w^{(1)}_5$};
\node[below] at (C1) {$w^{(2)}_1$};\node[below] at (C2) {$w^{(2)}_2$};\node[below] at (C3) {$w^{(2)}_3$};\node[below] at (C4) {$w^{(2)}_4$};\node[below] at (C5) 
{$w^{(2)}_5$};
\node[below] at (D1) {$w^{(12)}_1$};\node[below] at (D2) {$w^{(12)}_2$};\node[below] at (D3) {$w^{(12)}_3$};\node[below] at (D4) {$w^{(12)}_4$};\node[below] at (D5) 
{$w^{(12)}_5$};
\node[right] at (2.5,1.3) {$A$}; \node[right] at (14.7,-1.1) {$B$}; 
\node[above] at (-2.6,-.7) {$W^{(1)}$};
\node[above] at (5.5,-0.7) {$W^{(2)}$};
\node[above] at (12.1,-0.7) {$W^{(12)}$};
\node at (8.7,-1.1){$\cdots$};
\end{tikzpicture}
\caption{The construction for the proof of ~\Cref{thm:eternalRomanDominationNPHard} for graph $G$ of~\Cref{fig:EDS}. 
In the auxiliary split graph $H=(A,B;F)$, the edges between $A$ and $W^{(1)}$ are shown in the figure, while the edges between $A$ and every other $W^{(j)}$, $j\in [2,12]$, are  dashed lines and those within $A$ are omitted. 
The color of the vertices represents a configuration of an $m$-eternal Roman dominating function of $H$: two, one, and zero guards are placed on black, grey, and white vertices, respectively.
}
\label{fig:RDS}  
\end{figure}
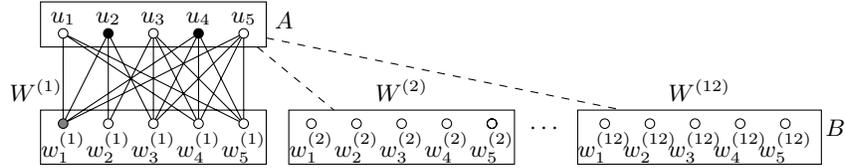

The proof of our next hardness result uses a construction in the same style as the one used for proving~\Cref{thm:eternalRomanDominationNPHard}. 

\begin{restatable}{lemma}{connectedIcore}
\label{lem:connectedIcore}%
    Let $G$ be a graph and $f$ be an Italian dominating function of $G$ such that the vertices with non-zero value induce a connected subgraph; moreover, let $t_f$ denote the weight of $f$. 
    Then $\gamma^\infty_I(G)\leq t_f+1$.
\end{restatable}

\begin{proof}
For every $v\in V$ such that $f(v)=0$, define the function $f_v$ as follows: $f_v(u)=f(u)$ if $u\neq v$ and $f_v(v)=1$.
Clearly, $f_v$ is an Italian dominating function of $G$ such that the vertices with non-zero values induce a connected subgraph, and the weight of $f_v$ is equal to $t_f+1$. 

At every step of the attacking sequence, we will maintain the following invariant: there are two guards on each vertex $u$ such that $f(u)=2$, one guard on each vertex $u$ such that $f(u)=1$ and one extra guard anywhere, which we will refer to as the {\it floating guard}, positioned on vertex $v$. 
Whenever any vertex $v'$ which does not have a guard is attacked, let $v_0,\ldots, v_k$ be any path such that $v_0=v$, $v_k=v'$ and $v_i\in Z$ for every $i= 1, \ldots, k-1$, for some $k\geq 1$. 
For every $i<k$, a guard on $v_i$ moves to $v_{i+1}$. After that, each vertex with a non-zero value of $f$ has the same number of guards as before the attack, and the floating guard is at vertex $v'$, thus both restoring the invariant and successfully defending the attack.

The weight of $f_v$ is equal to $t_f+1$ by construction and, since $G$ is finite, the weight is preserved through the sequences of Italian dominating functions defending the sequence of attacks.\qed  
\end{proof}

\begin{restatable}{theorem}{eternalItalianDominationNPHard}
    \label{thm:eternalItalianDominationNPHard}
     The \textsc{$m$-Eternal Italian Domination} and \textsc{$m$-Eternal Connected Italian Domination} problems are NP-hard, even on split graphs.
\end{restatable}
\begin{proof}
  We describe here a reduction from the \textsc{Italian Domination} problem. The construction is identical to the one used in the proof of~\Cref{thm:eternalRomanDominationNPHard}. The justification for the equivalence of instances is also analogous, but we make an explicit argument here for the sake of completeness. 

    Let \(\langle G = (V,E); k+1\rangle\) be an instance of \textsc{Italian Dominating Set} and let $V=\{v_1, \ldots, v_n\}$. Construct a split graph $H=(A\cup B,F)$ as follows: let $A=\{u_1,\ldots, u_n\}$ be a copy of $V$ and let $B=\{w_i^{(j)}$ for $1\leq i\leq n$ and $1\leq j\leq n+2$ contain $n+2$ copies of $V$. The edge set $F$ is defined as follows: \(A\) induces a clique and \(B\) induces an independent set; moreover, \(u_i\) is adjacent to \(w_k^{(j)}\) for all \(j\) iff \((v_i,v_k) \in E\) and \(u_i\) is adjacent to \(w_i^{(j)}\) for all \(i,j\).

Now consider the instance \(\langle H; k+1\rangle\) for the \textsc{$m$-Eternal Italian Domination Number}. We now argue the equivalence of these instances.

\paragraph{Forward Direction.} Let $f$ be an Italian dominating function of $G$ of weight \(k\). 
Let $T$ be the set of vertices in $A$ that corresponds to non-zero values of $f$, i.e. $u_i\in T$ if and only if $f(v_i)\neq 0$. The set $T$ induces a connected subgraph because it is a subset of $A$, which induces a clique in $H$.
By~\Cref{lem:connectedIcore}, $\gamma^\infty_I(H)\leq k+1$.

\paragraph{Reverse Direction.}
    Suppose that the $m$-eternal Italiam dominating number of $H$ is at most \(k+1\) and consider an initial configuration of $k+1$ guards such that $B$ contains at least one guard, that is an Italian dominating function $f$ of $H$ of weight equal to $k+1$. 
    By the pigeon-hole principle, there exists a \(1 \leqslant j \leqslant n+2\) such that \(W^{(j)}\) contains only zero value vertices, i.e., \(f(w_i^{(j)})=0\) for every \(i \leqslant n\). Let $f_G$ be the function defined on the vertices of $G$ as follows: $f_G(v_i)=f(u_i)$. Since \(W^{(j)}\) only has neighbors in \(A\), then $f_G$ is an Italian dominating function of $G$ of value at most $k$.
    
    Finally, we consider the instance \(\langle H;k+1\rangle\) for the \textsc{$m$-Eternal Connected Italian Domination}. To prove the equivalence of these instances, we re-use the previous proof and note that all considered configurations involved in a winning strategy are connected.\qed   
\end{proof}

\section{$m$-Eternal Domination on Infinite Regular Grids}
\label{ss:inf}

In this section, we provide a starting configuration and a strategy for the $m$-domination of infinite regular grids. We show that in all cases, we are able to obtain strongly optimal dominating sets that are also $m$-eternal dominating sets. We start with infinite square grids first. 

\begin{restatable}[$m$-Eternal Domination on the infinite Square Grid]{theorem}{squareGrid}
\label{thm:squareGrid}
There is a strongly optimal dominating set for the infinite square grid $T_4$ that is also an $m$-eternal dominating set. 
\end{restatable}
\begin{proof}
We define vertex set $S_4$ on the infinite square grid $T_4$ as follows: $(0,0)\in S_4$; moreover, if $(x,y)\in S_4$ then also $(x+2,y+1)$, $(x-1,y+2)$, $(x-2,y-1)$ and $(x+1,y-2)$ belong to $S_4$. This vertex set can be visualized in the left graph in~\Cref{fig:T48}. 

Note that the infinite set of closed neighborhoods $N[s]$ with $s \in S_4$ forms a partition of the vertices of $T_4$; therefore, $S_4$ is a dominating set of $T_4$ and we give a formal proof below.

Observe that for any guard at $(x,y)$ there is a guard at $(x', y')=(x+2,y+1)$, a guard at $(x'',y'')=(x'+1,y'-2)=(x+3, y-1)$ and a guard at $(x''', y''')=(x''+2,y''+1)=(x+5,y)$; hence, looking at the rows of $T_4$, for each guard at $(x,y)$, the four vertices to its right do not contain any guard. 
Nevertheless, $(x+1,y)$ is dominated by the guard at $(x,y)$; $(x+2, y)$ is dominated by the guard at $(x', y')=(x+2, y+1)$; $(x+3, y)$ is dominated by the guard at $(x'',y'')=(x+3, y-1)$; $(x+4, y)$ is dominated by $(x''', y''')=(x+5, y)$.
The generality of the reasoning proves that
$S_4$ is a dominating set of $T_4$.

Moreover, by construction, it holds $|N[v]\cap S_4|=1$ for every $v\in V$: this means that $S_4$ is strongly optimal. Next, we show that $S_4$ is a configuration of an $m$-eternal dominating set of $T_4$. Indeed, consider an attack on a vertex $v\in V\setminus S_4$. Since $S_4$ is a dominating set and, by construction, there exists a unique vertex $s\in S_4$ that is a neighbor of $v$. It is not restrictive to consider the case $s=(i^*,j^*)$ and $v=(i^*+1,j^*)$ (the other three cases are analogous). For every $(i,j)\in S_4$ (including $(i^*,j^*)$), the guard in $(i,j)$ moves to $(i+1,j)$. The new position of the guards is a translation of $S_4$ by one unit in the same direction, and thus still forms a dominating set. Therefore, with this strategy, the guards move along configurations of an $m$-eternal dominating set.
\qed
\end{proof}

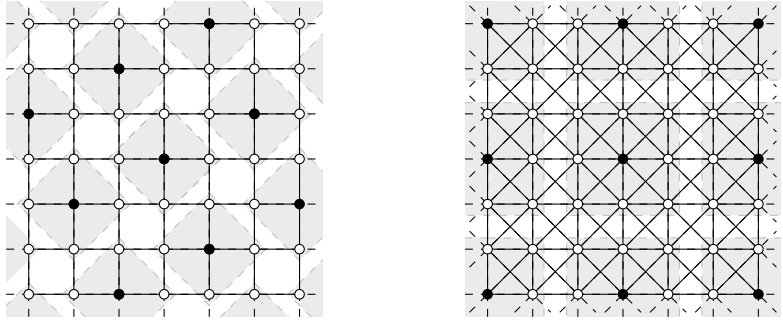
\begin{figure}[ht]
\centering
\hspace*{6mm}
\begin{minipage}{0.5\textwidth}
\begin{tikzpicture}[scale=0.6]
\clip (-3.5,-3.5) rectangle (3.5,3.5);
\fill[lightgray!32] (0,1.25) -- (1.25,0) -- (0,-1.25)-- (-1.25,0) -- cycle;
\draw[dashed, lightgray] (0,1.25) -- (1.25,0) -- (0,-1.25)-- (-1.25,0) -- cycle;
\fill[lightgray!32] (2,2.25) -- (3.25,1) -- (2,-.25)-- (.75,1) -- cycle;
\draw[dashed, lightgray] (2,2.25) -- (3.25,1) -- (2,-.25)-- (.75,1) -- cycle;
\fill[lightgray!32] (-1,3.25) -- (.25,2) -- (-1,.75)-- (-2.25,2) -- cycle;
\draw[dashed, lightgray](-1,3.25) -- (.25,2) -- (-1,.75)-- (-2.25,2) -- cycle;
\fill[lightgray!32] (-2,.25) -- (-.75,-1) -- (-2,-2.25)-- (-3.25,-1) -- cycle;
\draw[dashed, lightgray] (-2,.25) -- (-.75,-1) -- (-2,-2.25)-- (-3.25,-1) -- cycle;
\fill[lightgray!32] (1,-.75) -- (2.25,-2) -- (1,-3.25)-- (-.25,-2) -- cycle;
\draw[dashed, lightgray](1,-.75) -- (2.25,-2) -- (1,-3.25)-- (-.25,-2) -- cycle;
\fill[lightgray!32] (-3,2.25) -- (-1.75,1) -- (-3,-.25)-- (-4.25,1) -- cycle;
\draw[dashed, lightgray] (-3,2.25) -- (-1.75,1) -- (-3,-.25)-- (-4.25,1) -- cycle;
\fill[lightgray!32] (3,.25) -- (4.25,-1) -- (3,-2.25)-- (1.75,-1) -- cycle;
\draw[dashed, lightgray] (3,.25) -- (4.25,-1) -- (3,-2.25)-- (1.75,-1) -- cycle;
\fill[lightgray!32] (1,4.25) -- (2.25,3) -- (1,1.75)-- (-.25,3) -- cycle;
\draw[dashed, lightgray] (1,4.25) -- (2.25,3) -- (1,1.75)-- (-.25,3) -- cycle;
\fill[lightgray!32] (-1,-1.75) -- (.25,-3) -- (-1,-4.25)-- (-2.25,-3) -- cycle;
\draw[dashed, lightgray]  (-1,-1.75) -- (.25,-3) -- (-1,-4.25)-- (-2.25,-3) -- cycle;
\fill[lightgray!32] (4,3.25) -- (5.25,2) -- (4,.75)-- (2.75,2) -- cycle;
\draw[dashed, lightgray](4,3.25) -- (5.25,2) -- (4,.75)-- (2.75,2) -- cycle;
\fill[lightgray!32] (3,5.25) -- (4.25,4) -- (3,2.75)-- (1.75,4) -- cycle;
\draw[dashed, lightgray] (3,5.25) -- (4.25,4) -- (3,2.75)-- (1.75,4) -- cycle;
\fill[lightgray!32] (2,-2.75) -- (3.25,-4) -- (2,-5.25)-- (.75,-4) -- cycle;
\draw[dashed, lightgray] (2,-2.75) -- (3.25,-4) -- (2,-5.25)-- (.75,-4) -- cycle;
\fill[lightgray!32] (4,-1.75) -- (5.25,-3) -- (4,-4.25)-- (2.75,-3) -- cycle;
\draw[dashed, lightgray] (4,-1.75) -- (5.25,-3) -- (4,-4.25)-- (2.75,-3) -- cycle;
\fill[lightgray!32] (-4,-.75) -- (-2.75,-2) -- (-4,-3.25)-- (-5.25,-2) -- cycle;
\draw[dashed, lightgray] (0,1.25) -- (1.25,0) -- (0,-1.25)-- (-1.25,0) -- cycle;
\fill[lightgray!32] (-3,-2.75) -- (-1.75,-4) -- (-3,-5.25)-- (-4.25,-4) -- cycle;
\draw[dashed, lightgray]  (-3,-2.75) -- (-1.75,-4) -- (-3,-5.25)-- (-4.25,-4) -- cycle;
\fill[lightgray!32] (-2,5.25) -- (-.75,4) -- (-2,2.75)-- (-3.25,4) -- cycle;
\draw[dashed, lightgray] (-2,5.25) -- (-.75,4) -- (-2,2.75)-- (-3.25,4) -- cycle;
\fill[lightgray!32] (-4,4.25) -- (-2.75,3) -- (-4,1.75)-- (-5.25,3) -- cycle;
\draw[dashed, lightgray] (-4,4.25) -- (-2.75,3) -- (-4,1.75)-- (-5.25,3) -- cycle;
\draw(-3,-3)--(-3,3)(-2,-3)--(-2,3)(-1,-3)--(-1,3)(0,-3)--(0,3)(1,-3)--(1,3)(2,-3)--(2,3)(3,-3)--(3,3)(-3,-3)--(3,-3)(-3,-2)--(3,-2)(-3,-1)--(3,-1)(-3,0)--(3,0)(-3,1)--(3,1)(-3,2)--(3,2)(-3,3)--(3,3);
\draw[dashed] (-3,-3.5)--(-3,3.5)(-2,-3.5)--(-2,3.5)(-1,-3.5)--(-1,3.5)(0,-3.5)--(0,3.5)(1,-3.5)--(1,3.5)(2,-3.5)--(2,3.5)(3,-3.5)--(3,3.5)(-3.5,-3)--(3.5,-3)(-3.5,-2)--(3.5,-2)(-3.5,-1)--(3.5,-1)(-3.5,0)--(3.5,0)(-3.5,1)--(3.5,1)(-3.5,2)--(3.5,2)(-3.5,3)--(3.5,3);
\draw[fill=white]
(-3,-3)circle [radius=3pt]
(-3,-2) circle [radius=3pt]
(-3,-1) circle [radius=3pt]
(-3,0) circle [radius=3pt]
(-3,2) circle [radius=3pt]
(-3,3)circle [radius=3pt]
(-2,-3) circle [radius=3pt]
(-2,-2) circle [radius=3pt]
(-2,0) circle [radius=3pt]
(-2,1) circle [radius=3pt]
(-2,2) circle [radius=3pt]
(-2,3) circle [radius=3pt]
(-1,-3) circle [radius=3pt]
(-1,-2) circle [radius=3pt]
(-1,-1) circle [radius=3pt]
(-1,0) circle [radius=3pt]
(-1,1) circle [radius=3pt]
(-1,2) circle [radius=3pt]
(-1,3) circle [radius=3pt]
(0,-3) circle [radius=3pt]
(0,-2) circle [radius=3pt]
(0,-1) circle [radius=3pt]
(0,0) circle [radius=3pt]
(0,1) circle [radius=3pt]
(0,2) circle [radius=3pt]
(0,3) circle [radius=3pt]
(1,-3) circle [radius=3pt]
(1,-2) circle [radius=3pt]
(1,-1) circle [radius=3pt]
(1,0) circle [radius=3pt]
(1,1) circle [radius=3pt]
(1,2) circle [radius=3pt]
(1,3) circle [radius=3pt]
(2,-3) circle [radius=3pt]
(2,-2) circle [radius=3pt]
(2,-1) circle [radius=3pt]
(2,0) circle [radius=3pt]
(2,1) circle [radius=3pt]
(2,2) circle [radius=3pt]
(2,3) circle [radius=3pt]
(3,-3) circle [radius=3pt]
(3,-2) circle [radius=3pt]
(3,0) circle [radius=3pt]
(3,1) circle [radius=3pt]
(3,2) circle [radius=3pt]
(3,3) circle [radius=3pt];
\draw[fill=black]
(0,0) circle [radius=3pt]
(2,1) circle [radius=3pt]
(-1,2) circle [radius=3pt]
(-2,-1) circle [radius=3pt]
(1,-2) circle [radius=3pt]
(-3,1) circle [radius=3pt]
(3,-1) circle [radius=3pt]
(1,3) circle [radius=3pt]
(-1,-3) circle [radius=3pt];
\end{tikzpicture}
\end{minipage}%
\begin{minipage}{0.5\textwidth}
\begin{tikzpicture}[scale=0.6]
\clip (-3.5,-3.5) rectangle (3.5,3.5);
\fill[lightgray!32] (0,1.25) -- (1.25, 1.25) -- (1.25,0) -- (1.25, -1.25) -- (0,-1.25) -- (-1.25, -1.25) -- (-1.25,0) -- (-1.25, 1.25) -- cycle;
\draw[dashed, lightgray]  (0,1.25) -- (1.25, 1.25) -- (1.25,0) -- (1.25, -1.25) -- (0,-1.25) -- (-1.25, -1.25) -- (-1.25,0) -- (-1.25, 1.25) -- cycle;
\fill[lightgray!32] (3,1.25) -- (4.25, 1.25) -- (4.25,0) -- (4.25, -1.25) -- (3,-1.25) -- (1.75, -1.25) -- (1.75,0) -- (1.75, 1.25) -- cycle;
\draw[dashed, lightgray] (3,1.25) -- (4.25, 1.25) -- (4.25,0) -- (4.25, -1.25) -- (3,-1.25) -- (1.75, -1.25) -- (1.75,0) -- (1.75, 1.25) -- cycle;
\fill[lightgray!32] (-3,1.25) -- (-1.75, 1.25) -- (-1.75,0) -- (-1.75, -1.25) -- (-3,-1.25) -- (-4.25, -1.25) -- (-4.25,0) -- (-4.25, 1.25) -- cycle;
\draw[dashed, lightgray] (-3,1.25) -- (-1.75, 1.25) -- (-1.75,0) -- (-1.75, -1.25) -- (-3,-1.25) -- (-4.25, -1.25) -- (-4.25,0) -- (-4.25, 1.25) -- cycle;
\fill[lightgray!32] (0,4.25) -- (1.25, 4.25) -- (1.25,3) -- (1.25, 1.75) -- (0,1.75) -- (-1.25, 1.75) -- (-1.25,3) -- (-1.25, 4.25) -- cycle;
\draw[dashed, lightgray]  (0,4.25) -- (1.25, 4.25) -- (1.25,3) -- (1.25, 1.75) -- (0,1.75) -- (-1.25, 1.75) -- (-1.25,3) -- (-1.25, 4.25) -- cycle;
\fill[lightgray!32] (3,4.25) -- (4.25, 4.25) -- (4.25,3) -- (4.25, 1.75) -- (3,1.75) -- (1.75, 1.75) -- (1.75,3) -- (1.75, 4.25) -- cycle;
\draw[dashed, lightgray]  (3,4.25) -- (4.25, 4.25) -- (4.25,3) -- (4.25, 1.75) -- (3,1.75) -- (1.75, 1.75) -- (1.75,3) -- (1.75, 4.25) -- cycle;
\fill[lightgray!32] (-3,4.25) -- (-1.75, 4.25) -- (-1.75,3) -- (-1.75, 1.75) -- (-3,1.75) -- (-4.25, 1.75) -- (-4.25,3) -- (-4.25, 4.25) -- cycle;
\draw[dashed, lightgray] (-3,4.25) -- (-1.75, 4.25) -- (-1.75,3) -- (-1.75, 1.75) -- (-3,1.75) -- (-4.25, 1.75) -- (-4.25,3) -- (-4.25, 4.25) -- cycle;
\fill[lightgray!32] (0,-1.75) -- (1.25, -1.75) -- (1.25,-3) -- (1.25, -4.25) -- (0,-4.25) -- (-1.25, -4.25) -- (-1.25,-3) -- (-1.25, -1.75) -- cycle;
\draw[dashed, lightgray]  (0,-1.75) -- (1.25, -1.75) -- (1.25,-3) -- (1.25, -4.25) -- (0,-4.25) -- (-1.25, -4.25) -- (-1.25,-3) -- (-1.25, -1.75) -- cycle;
\fill[lightgray!32] (3,-1.75) -- (4.25, -1.75) -- (4.25,-3) -- (4.25, -4.25) -- (3,-4.25) -- (1.75, -4.25) -- (1.75,-3) -- (1.75, -1.75) -- cycle;
\draw[dashed, lightgray]  (3,-1.75) -- (4.25, -1.75) -- (4.25,-3) -- (4.25, -4.25) -- (3,-4.25) -- (1.75, -4.25) -- (1.75,-3) -- (1.75, -1.75) -- cycle;
\fill[lightgray!32] (-3,-1.75) -- (-1.75, -1.75) -- (-1.75,-3) -- (-1.75, -4.25) -- (-3,-4.25) -- (-4.25, -4.25) -- (-4.25,-3) -- (-4.25, -1.75) -- cycle;
\draw[dashed, lightgray] (-3,-1.75) -- (-1.75, -1.75) -- (-1.75,-3) -- (-1.75, -4.25) -- (-3,-4.25) -- (-4.25, -4.25) -- (-4.25,-3) -- (-4.25, -1.75) -- cycle;
\draw
(-3,-3)--(-3,3)
(-2,-3)--(-2,3)
(-1,-3)--(-1,3)
(0,-3)--(0,3)
(1,-3)--(1,3)
(2,-3)--(2,3)
(3,-3)--(3,3)
(-3,-3)--(3,-3)
(-3,-2)--(3,-2)
(-3,-1)--(3,-1)
(-3,0)--(3,0)
(-3,1)--(3,1)
(-3,2)--(3,2)
(-3,3)--(3,3)
(-3,-2)--(-2,-3)
(-3,-1)--(-1,-3)
(-3,0)--(0,-3)
(-3,1)--(1,-3)
(-3,2)--(2,-3)
(-3,3)--(3,-3)
(-2,3)--(3,-2)
(-1,3)--(3,-1)
(0,3)--(3,0)
(1,3)--(3,1)
(2,3)--(3,2)
(-3,2)--(-2,3)
(-3,1)--(-1,3)
(-3,0)--(0,3)
(-3,-1)--(1,3)
(-3,-2)--(2,3)
(-3,-3)--(3,3)
(-2,-3)--(3,2)
(-1,-3)--(3,1)
(0,-3)--(3,0)
(1,-3)--(3,-1)
(2,-3)--(3,-2);
\draw[dashed] 
(-3,-3.5)--(-3,3.5)
(-2,-3.5)--(-2,3.5)
(-1,-3.5)--(-1,3.5)
(0,-3.5)--(0,3.5)
(1,-3.5)--(1,3.5)
(2,-3.5)--(2,3.5)
(3,-3.5)--(3,3.5)
(-3.5,-3)--(3.5,-3)
(-3.5,-2)--(3.5,-2)
(-3.5,-1)--(3.5,-1)
(-3.5,0)--(3.5,0)
(-3.5,1)--(3.5,1)
(-3.5,2)--(3.5,2)
(-3.5,3)--(3.5,3)
(-3.4,1.6)--(-1.6,3.4)
(-3.4,0.6)--(-0.6,3.4)
(-3.4,-0.4)--(0.4,3.4)
(-3.4,-1.4)--(1.4,3.4)
(-3.4,-2.4)--(2.4,3.4)
(-3.4,-3.4)--(3.4,3.4)
(-2.4,-3.4)--(3.4,2.4)
(-1.4,-3.4)--(3.4,1.4)
(-0.4,-3.4)--(3.4,0.4)
(0.6,-3.4)--(3.4,-0.6)
(1.6,-3.4)--(3.4,-1.6)
(-3.4,-1.6)--(-1.6,-3.4)
(-3.4,-0.6)--(-0.6,-3.4)
(-3.4,0.4)--(0.4,-3.4)
(-3.4,1.4)--(1.4,-3.4)
(-3.4,2.4)--(2.4,-3.4)
(-3.4,3.4)--(3.4,-3.4)
(-2.4,3.4)--(3.4,-2.4)
(-1.4,3.4)--(3.4,-1.4)
(-0.4,3.4)--(3.4,-0.4)
(0.6,3.4)--(3.4,0.6)
(1.6,3.4)--(3.4,1.6);
\draw[fill=white]
(-3,-2) circle [radius=3pt]
(-3,-1) circle [radius=3pt]
(-3,1) circle [radius=3pt]
(-3,2) circle [radius=3pt]
(-2,-3) circle [radius=3pt]
(-2,-2) circle [radius=3pt]
(-2,-1) circle [radius=3pt]
(-2,0) circle [radius=3pt]
(-2,1) circle [radius=3pt]
(-2,2) circle [radius=3pt]
(-2,3) circle [radius=3pt]
(-1,-3) circle [radius=3pt]
(-1,-2) circle [radius=3pt]
(-1,-1) circle [radius=3pt]
(-1,0) circle [radius=3pt]
(-1,1) circle [radius=3pt]
(-1,2) circle [radius=3pt]
(-1,3) circle [radius=3pt]
(0,-2) circle [radius=3pt]
(0,-1) circle [radius=3pt]
(0,1) circle [radius=3pt]
(0,2) circle [radius=3pt]
(1,-3) circle [radius=3pt]
(1,-2) circle [radius=3pt]
(1,-1) circle [radius=3pt]
(1,0) circle [radius=3pt]
(1,1) circle [radius=3pt]
(1,2) circle [radius=3pt]
(1,3) circle [radius=3pt]
(2,-3) circle [radius=3pt]
(2,-2) circle [radius=3pt]
(2,-1) circle [radius=3pt]
(2,0) circle [radius=3pt]
(2,1) circle [radius=3pt]
(2,2) circle [radius=3pt]
(2,3) circle [radius=3pt]
(3,-2) circle [radius=3pt]
(3,-1) circle [radius=3pt]
(3,1) circle [radius=3pt]
(3,2) circle [radius=3pt];
\draw[fill=black]
(0,-3) circle [radius=3pt]
(0,3) circle [radius=3pt]
(-3,0) circle [radius=3pt]
(3,0) circle [radius=3pt]
(0,0) circle [radius=3pt]
(3,3) circle [radius=3pt]
(3,-3) circle [radius=3pt]
(-3,-3) circle [radius=3pt]
(-3,3) circle [radius=3pt];
\end{tikzpicture}
\end{minipage}
\caption{On the left, graph $T_4$ and, on the right, graph $T_8$. 
In both figures, the black vertices represent the dominating set $S$ described in the proof of Theorems~\ref{thm:squareGrid} and of~\ref{thm:octagonalGrid}, respectively, while the shadowed zones represent how vertices of $S$ dominate their neighbors, thus creating a partition of the vertices of the graph.}
\label{fig:T48}  
\end{figure}


We next turn to infinite octagonal grids. The following theorem shows that we can come up with a strongly optimal dominating set that is also an eternal dominating set. The argument here is similar in spirit to square grids. 
The proposed solution can be visualized in the right graph in~\Cref{fig:T48}.

\begin{restatable}[$m$-Eternal Domination on the infinite Octagonal Grid]{theorem}{octagonalGrid}
    \label{thm:octagonalGrid}
    There is a strongly optimal dominating set for the infinite octagonal grid $T_8$ that is also an $m$-eternal dominating set. 
\end{restatable}
\begin{proof}
  We define vertex set $S$ on the infinite octagonal grid $T_8$ as follows: $(0,0)\in S$; moreover, if $(x,y)\in S$ then also $(x,y+3)$,$(x,y-3)$, $(x+3,y)$ and $(x-3,y)$ belong to $S$. 
    This vertex set can be visualized in~\Cref{fig:T48}. Note that the infinite set of closed neighborhoods $N[S]$ of each vertex $s\in S$ is a partition of the vertices of $T_8$; therefore, 
    $S$ is a dominating set of $T_8$. Indeed, observe that for any guard at $(x,y)$ there is a guard at $(x',y')=(x+3,y)$ and no guards lie at $y$-coordinate $y+1$ and $y-1$. 
    Hence, looking at the rows of $T_8$, for each guard at $(x,y)$, the two vertices to its right do not contain any guard;
    nevertheless, $(x+1,y)$ is dominated by the guard at $(x,y)$ and $(x+2, y)$ is dominated by the guard at $(x+3, y)$.
    Moreover, each vertex at $y$-coordinate $y+1$ is dominated, indeed they can be partitioned into 3 vertex sets $(x-1, y+1)$, $(x, y+1)$ and $(x+1, y+1)$ dominated by the guard at $(x,y)$.
    Finally, in the same way, each vertex at $y$-coordinate $y+2$ is dominated.
    The generality of the reasoning proves that
    $S$ is a dominating set of $T_8$.

    Next, we show that $S$ is a configuration of an $m$-eternal dominating set of $T_8$. 
    Indeed, consider an attack on a vertex $v\in V\setminus S$. 
    Since $S$ is a dominating set and by construction, there exists a unique vertex $s\in S$ that is a neighbor of $v$. 
    It is not restrictive to consider the case $s=(i^*,j^*)$ and $v=(i^*+1,j^*)$ (the other seven cases are analogous). 
    For every $(i,j)\in S$ (included $(i^*,j^*)$), the guard on $(i,j)$ moves to $(i+1,j)$. 
    The new position of the guards is a translation of $S$ by one unit in the same direction and thus still forms a dominating set. 
    Therefore, with this strategy, the guards move along configurations of an $m$-eternal dominating set.
    
     We have that $S$ is strongly optimal because for every $s\in S$, $S$ uses only one guard for the nine vertices of $N[s]$, formally 
     $|N[s]\cap S|=1$, and $\{N[s]~|~s\in S\}$ is a partition of $T_8$.
     \qed  
\end{proof}

Also for the infinite hexagonal and triangular grids we obtain strongly optimal eternal dominating sets for both cases. 

\begin{restatable}[$m$-Eternal Domination on the infinite Hexagonal Grid]{theorem}{hexagonalGrid} 
\label{thm:hexagonalGrid}
There is a strongly optimal dominating set for the infinite hexagonal grid $T_3$ that is also an eternal dominating set. 
\end{restatable}
\begin{proof}
   We define vertex set $S$ on the infinite hexagonal grid $T_3$ as follows: $(0,0)\in S$; moreover, if $(x,y)\in S$ then also $(x+2,y+2)$, $(x+3,y-1)$, $(x+1,y-3)$, $(x-2,y-2)$, $(x-3,y+1)$ and $(x-1,y+3)$ belong to $S$. Finally, if 
    $x+y$ is even, also $(x-1,y)\in S$. 
    This vertex set can be visualized in~\Cref{fig:T3}. 
    For any two integers $x,y\in \mathbb{Z}$, we say that $(x,y)$ is a \rp vertex if $x+y$ is even and a \lp vertex otherwise. 
    Let $S_{\rp}$ and $S_{\lp}$ the subsets of $S$ constituted by the \rp vertices and \lp vertices of $S$, respectively; note that these two sets partition $S$.
    Note that the infinite set of closed neighborhoods $N[S]$ of each $s\in S$ is a partition of the vertices of $T_3$; therefore, $S$ is a dominating set of $T_3$.

    Observe that for any guard at $(x,y)$ such that $x+y$ is even, there is a guard at $(x', y')=(x-1, y)$, a guard at $(x'', y'')=(x-1, y+3)$, a guard at $(x''', y''')=(x'+2, y'+2)=(x+1, y+2)$, a guard at $(x^{iv}, y^{iv})=(x'''-1, y'''+3)=(x, y+5)$; 
    moreover, there is a guard at $(x^v, y^v)=(x+2, y+2)$, a guard at $(x^{vi}, y^{vi})=(x^v+3, y^v-1)=(x+1, y+5)$, and a guard at $(x^{vii}, y^{vii})=(x^{vi}-1, y^{vi}+3)=(x, y+8)$. 
    Hence, looking at the rows of $T_3$, for each guard at $(x,y)$, the 4 vertices above it do not contain any guard, the fifth does, and the next 2 vertices above do not contain a guard again.
    Nevertheless, $(x+1, y)$ is dominated by the guard at $(x, y)$; 
    $(x+2, y)$ is dominated by the guard at $(x''', y''')=(x+1, y+2)$;
    $(x+3, y)$ is dominated by the guard at $(x'', y'')=(x-1, y+3)$;
    $(x+4, y)$ and $(x+6, y)$ are dominated by the guard at $(x^{iv}, y^{iv})=(x, y+5)$;
    finally, $(x+7, y)$ is dominated by the guard at $(x^{vii}, y^{vii})=(x,y+8)$.
    Analogous considerations can be done starting from $(x,y)$ such that $x+y$ is odd.
    The generality of the reasoning proves that
    $S$ is a dominating set of $T_3$.

    $S$ is a configuration of an $m$-eternal dominating set of $T_3$. 
    Indeed, 
    consider an attack to a vertex $v\in V\setminus S$. 
    Since $S$ is a dominating set and by construction, there exists a unique vertex $s\in S$ that is a neighbor of $v$. 
    Assume first that $s=(i^*,j^*)$ belongs to $S_{\rp}$. 
    If $v=(i^*,j^*+1)$, then the guards moves as follows: for every $(i,j)\in S_{\rp}$, the guard on $(i,j)$ (included $(i^*,j^*)$) moves to $(i,j+1)$ and for every $(i,j)\in S_{\lp}$, the guard on $(i,j)$ move to $(i,j-1)$. 
    A display of this movement can be visualized in~\Cref{fig:T3}. 
    If $v=(i^*+1,j^*)$, then the guards moves as follows: for every $(i,j)\in S_{\rp}$, the guard on $(i,j)$  (included $(i^*,j^*)$) moves to $(i+1,j)$ and for every $(i,j)\in S_{\lp}$, the guard on $(i,j)$ move to $(i,j+1)$.
    If $v=(i^*,j^*-1)$, then the guards moves as follows: for every $(i,j)\in S_{\rp}$, the guard on $(i,j)$  (included $(i^*,j^*)$) moves to $(i,j-1)$ and for every $(i,j)\in S_{\lp}$, the guard on $(i,j)$ move to $(i-1,j)$. 
    The case where $s=(i^*,j^*)$ belongs to $S_{\lp}$ is dealt symmetrically. 
    The new position of the guards is a translation of $S$ by one unit in the same direction and thus still forms a dominating set. 
    Therefore, with this strategy, the guards move along configurations of an $m$-eternal dominating set.
    
    We have that $S$ is strongly optimal because for every $s\in S$, $S$ uses only one guard for the four vertices of $N[s]$, formally 
    $|N[s]\cap S|=1$, and $\{N[s]~|~s\in S\}$ is a partition of $T_3$. \qed   
\end{proof}

\begin{figure}[ht]
\centering
\vspace*{-.2cm}
\hspace*{1cm}
\begin{minipage}{0.5\textwidth}
\begin{tikzpicture}[scale=0.6]
\clip (-3.5,-3.5) rectangle (3.5,3.5);
\fill[lightgray!32] (-0.8,1.5) -- (-.8,-1.5) -- (-2.5,0) -- cycle;
\draw[dashed, lightgray] (-0.8,1.5) -- (-.8,-1.5) -- (-2.5,0) -- cycle;
\fill[lightgray!32] (-1.8,4.5) -- (-1.8,1.5) -- (-3.5,3) -- cycle;
\draw[dashed, lightgray](-1.8,4.5) -- (-1.8,1.5) -- (-3.5,3) -- cycle;
\fill[lightgray!32] (-2.8,-.5) -- (-2.8,-3.5) -- (-4.5,-2) -- cycle;
\draw[dashed, lightgray](-2.8,-.5) -- (-2.8,-3.5) -- (-4.5,-2) -- cycle;
\fill[lightgray!32] (.2,-1.5) -- (.2,-4.5) -- (-1.5,-3) -- cycle;
\draw[dashed, lightgray](.2,-1.5) -- (.2,-4.5) -- (-1.5,-3) -- cycle;
\fill[lightgray!32] (1.2,3.5) -- (1.2,.5) -- (-.5,2) -- cycle;
\draw[dashed, lightgray] (1.2,3.5) -- (1.2,.5) -- (-.5,2) -- cycle;
\fill[lightgray!32] (2.2,.5) -- (2.2,-2.5) -- (.5,-1) -- cycle;
\draw[dashed, lightgray]  (2.2,.5) -- (2.2,-2.5) -- (.5,-1) -- cycle;
\fill[lightgray!32] (-3.2,2.5) -- (-3.2,-0.5) -- (-1.5,1) -- cycle;
\draw[dashed, lightgray]  (-3.2,2.5) -- (-3.2,-0.5) -- (-1.5,1) -- cycle;
\fill[lightgray!32] (-2.2,-.5) -- (-2.2,-3.5) -- (-0.5,-2) -- cycle;
\draw[dashed, lightgray] (-2.2,-.5) -- (-2.2,-3.5) -- (-0.5,-2) -- cycle;
\fill[lightgray!32] (-1.2,4.5) -- (-1.2,1.5) -- (0.5,3) -- cycle;
\draw[dashed, lightgray] (-1.2,4.5) -- (-1.2,1.5) -- (0.5,3) -- cycle;
\fill[lightgray!32] (-.2,1.5) -- (-.2,-1.5) -- (1.5,0) -- cycle;
\draw[dashed, lightgray]  (-.2,1.5) -- (-.2,-1.5) -- (1.5,0) -- cycle;
\fill[lightgray!32] (0.8,-1.5) -- (0.8,-4.5) -- (2.5,-3) -- cycle;
\draw[dashed, lightgray] (0.8,-1.5) -- (0.8,-4.5) -- (2.5,-3) -- cycle;
\fill[lightgray!32] (1.8,3.5) -- (1.8,.5) -- (3.5,2) -- cycle;
\draw[dashed, lightgray] (1.8,3.5) -- (1.8,.5) -- (3.5,2) -- cycle;
\fill[lightgray!32] (2.8,.5) -- (2.8,-2.5) -- (4.5,-1) -- cycle;
\draw[dashed, lightgray]  (2.8,.5) -- (2.8,-2.5) -- (4.5,-1) -- cycle;
\fill[lightgray!32] (4.2,2.5) -- (4.2,-.5) -- (2.5,1) -- cycle;
\draw[dashed, lightgray]  (4.2,2.5) -- (4.2,-.5) -- (2.5,1) -- cycle;
\fill[lightgray!32] (3.2,5.5) -- (3.2,2.5) -- (1.5,4) -- cycle;
\draw[dashed, lightgray]  (3.2,5.5) -- (3.2,2.5) -- (1.5,4) -- cycle;
\fill[lightgray!32] (4.2,-1.5) -- (4.2,-4.5) -- (2.5,-3) -- cycle;
\draw[dashed, lightgray]  (4.2,-1.5) -- (4.2,-4.5) -- (2.5,-3) -- cycle;
\draw
(-3,-3)--(-3,3)
(-2,-3)--(-2,3)
(-1,-3)--(-1,3)
(0,-3)--(0,3)
(1,-3)--(1,3)
(2,-3)--(2,3)
(3,-3)--(3,3)
(-3,-3)--(-2,-3)(-1,-3)--(0,-3)(1,-3)--(2,-3)
(-2,-2)--(-1,-2)(0,-2)--(1,-2)(2,-2)--(3,-2)
(-3,-1)--(-2,-1)(-1,-1)--(0,-1)(1,-1)--(2,-1)
(-2,0)--(-1,0)(0,0)--(1,0)(2,0)--(3,0)
(-3,1)--(-2,1)(-1,1)--(0,1)(1,1)--(2,1)
(-2,2)--(-1,2)(0,2)--(1,2)(2,2)--(3,2)
(-3,3)--(-2,3)(-1,3)--(0,3)(1,3)--(2,3);
\draw[dashed] 
(-3,-3.5)--(-3,3.5)
(-2,-3.5)--(-2,3.5)
(-1,-3.5)--(-1,3.5)
(0,-3.5)--(0,3.5)
(1,-3.5)--(1,3.5)
(2,-3.5)--(2,3.5)
(3,-3.5)--(3,3.5)
(3,-3)--(3.5,-3)
(-3.5,-2)--(-3,-2)
(3,-1)--(3.5,-1)
(-3.5,0)--(-3,0)
(3,1)--(3.5,1)
(-3.5,2)--(-3,2)
(3,3)--(3.5,3);
\draw[very thick,->](0,0)--(1,0);
\draw[very thick,->](2,2)--(3,2);
\draw[very thick,->](-1,3)--(0,3);
\draw[very thick,->](-3,1)--(-2,1);
\draw[very thick,->](-2,-2)--(-1,-2);
\draw[very thick,->](1,-3)--(2,-3);
\draw[very thick,->,dashed](3,-1)--(3.5,-1);
\draw[very thick,->](2,-1)--(2,0);
\draw[very thick,->](1,2)--(1,3);
\draw[very thick,->,dashed](-2,3)--(-2,3.5);
\draw[very thick,->](-1,0)--(-1,1);
\draw[very thick,->](-3,-2)--(-3,-1);
\draw[very thick,->](0,-3)--(0,-2);
\draw[very thick,->,dashed](3,-3.5)--(3,-3.2);
\draw[fill=white]
(-3,-3) circle [radius=3pt]
(-3,-1) circle [radius=3pt]
(-3,0) circle [radius=3pt]
(-3,2) circle [radius=3pt]
(-3,3) circle [radius=3pt]
(-2,-3) circle [radius=3pt]
(-2,-1) circle [radius=3pt]
(-2,0) circle [radius=3pt]
(-2,1) circle [radius=3pt]
(-2,2) circle [radius=3pt]
(-1,-3) circle [radius=3pt]
(-1,-2) circle [radius=3pt]
(-1,-1) circle [radius=3pt]
(-1,1) circle [radius=3pt]
(-1,2) circle [radius=3pt]
(0,-2) circle [radius=3pt]
(0,-1) circle [radius=3pt]
(0,1) circle [radius=3pt]
(0,2) circle [radius=3pt]
(0,3) circle [radius=3pt]
(1,-2) circle [radius=3pt]
(1,-1) circle [radius=3pt]
(1,1) circle [radius=3pt]
(1,3) circle [radius=3pt]
(2,-3) circle [radius=3pt]
(2,-2) circle [radius=3pt]
(2,0) circle [radius=3pt]
(2,1) circle [radius=3pt]
(2,3) circle [radius=3pt]
(3,-3) circle [radius=3pt]
(3,-2) circle [radius=3pt]
(3,0) circle [radius=3pt]
(3,1) circle [radius=3pt]
(3,2) circle [radius=3pt]
(3,3) circle [radius=3pt];
\draw[fill=black]
(0,0) circle [radius=3pt]
(-1,0) circle [radius=3pt]
(2,2) circle [radius=3pt]
(1,2) circle [radius=3pt]
(2,-1) circle [radius=3pt]
(3,-1) circle [radius=3pt]
(0,-3) circle [radius=3pt]
(1,-3) circle [radius=3pt]
(-2,-2) circle [radius=3pt]
(-3,-2) circle [radius=3pt]
(-3,1) circle [radius=3pt]
(-2,3) circle [radius=3pt]
(-1,3) circle [radius=3pt];
\draw[fill=red]
(1,0) circle [radius=3pt];
\end{tikzpicture}
\end{minipage}%
\begin{minipage}{0.5\textwidth}
\begin{tikzpicture}[scale=0.6]
\clip (-3.5,-3.5) rectangle (3.5,3.5);
\fill[lightgray!32] (1.2,1.5) -- (1.2,-1.5) -- (-0.5,0) -- cycle;
\draw[dashed, lightgray] (1.2,1.5) -- (1.2,-1.5) -- (-0.5,0) -- cycle;
\fill[lightgray!32] (3.2,3.5) -- (3.2,.5) -- (1.5,2) -- cycle;
\draw[dashed, lightgray]  (3.2,3.5) -- (3.2,.5) -- (1.5,2) -- cycle;
\fill[lightgray!32] (2.2,-1.5) -- (2.2,-4.5) -- (.5,-3) -- cycle;
\draw[dashed, lightgray]  (2.2,-1.5) -- (2.2,-4.5) -- (.5,-3) -- cycle;
\fill[lightgray!32] (-.8,-.5) -- (-.8,-3.5) -- (-2.5,-2) -- cycle;
\draw[dashed, lightgray] (-.8,-.5) -- (-.8,-3.5) -- (-2.5,-2) -- cycle;
\fill[lightgray!32] (.2,4.5) -- (.2,1.5) -- (-1.5,3) -- cycle;
\draw[dashed, lightgray] (.2,4.5) -- (.2,1.5) -- (-1.5,3) -- cycle;
\fill[lightgray!32] (-1.2,2.5) -- (-1.2,-0.5) -- (.5,1) -- cycle;
\draw[dashed, lightgray] (-1.2,2.5) -- (-1.2,-0.5) -- (.5,1) -- cycle;
\fill[lightgray!32] (-.2,-.5) -- (-0.2,-3.5) -- (1.5,-2) -- cycle;
\draw[dashed, lightgray]  (-.2,-.5) -- (-0.2,-3.5) -- (1.5,-2) -- cycle;
\fill[lightgray!32] (.8,4.5) -- (.8,1.5) -- (2.5,3) -- cycle;
\draw[dashed, lightgray] (.8,4.5) -- (.8,1.5) -- (2.5,3) -- cycle;
\fill[lightgray!32] (1.8,1.5) -- (1.8,-1.5) -- (3.5,0) -- cycle;
\draw[dashed, lightgray] (1.8,1.5) -- (1.8,-1.5) -- (3.5,0) -- cycle;
\fill[lightgray!32] (2.8,-1.5) -- (2.8,-4.5) -- (4.5,-3) -- cycle;
\draw[dashed, lightgray] (2.8,-1.5) -- (2.8,-4.5) -- (4.5,-3) -- cycle;
\fill[lightgray!32] (4.2,.5) -- (4.2,-2.5) -- (2.5,-1) -- cycle;
\draw[dashed, lightgray](4.2,.5) -- (4.2,-2.5) -- (2.5,-1) -- cycle;
\fill[lightgray!32] (-3.2,.5) -- (-3.2,-2.5) -- (-1.5,-1) -- cycle;
\draw[dashed, lightgray] (-3.2,.5) -- (-3.2,-2.5) -- (-1.5,-1) -- cycle;
\fill[lightgray!32] (-3.2,4.5) -- (-3.2,1.5) -- (-1.5,3) -- cycle;
\draw[dashed, lightgray] (-3.2,4.5) -- (-3.2,1.5) -- (-1.5,3) -- cycle;
\fill[lightgray!32] (-1.8,2.5) -- (-1.8,-.5) -- (-3.5,1) -- cycle;
\draw[dashed, lightgray] (-1.8,2.5) -- (-1.8,-.5) -- (-3.5,1) -- cycle;
\fill[lightgray!32] (-2.2,-2.5) -- (-2.2,-5.5) -- (-0.5,-4) -- cycle;
\draw[dashed, lightgray](-2.2,-2.5) -- (-2.2,-5.5) -- (-0.5,-4) -- cycle;
\fill[lightgray!32] (-2.8,-2.5) -- (-2.8,-5.5) -- (-4.5,-4) -- cycle;
\draw[dashed, lightgray](-2.8,-2.5) -- (-2.8,-5.5) -- (-4.5,-4) -- cycle;
\draw
(-3,-3)--(-3,3)
(-2,-3)--(-2,3)
(-1,-3)--(-1,3)
(0,-3)--(0,3)
(1,-3)--(1,3)
(2,-3)--(2,3)
(3,-3)--(3,3)
(-3,-3)--(-2,-3)(-1,-3)--(0,-3)(1,-3)--(2,-3)
(-2,-2)--(-1,-2)(0,-2)--(1,-2)(2,-2)--(3,-2)
(-3,-1)--(-2,-1)(-1,-1)--(0,-1)(1,-1)--(2,-1)
(-2,0)--(-1,0)(0,0)--(1,0)(2,0)--(3,0)
(-3,1)--(-2,1)(-1,1)--(0,1)(1,1)--(2,1)
(-2,2)--(-1,2)(0,2)--(1,2)(2,2)--(3,2)
(-3,3)--(-2,3)(-1,3)--(0,3)(1,3)--(2,3);
\draw[dashed] 
(-3,-3.5)--(-3,3.5)
(-2,-3.5)--(-2,3.5)
(-1,-3.5)--(-1,3.5)
(0,-3.5)--(0,3.5)
(1,-3.5)--(1,3.5)
(2,-3.5)--(2,3.5)
(3,-3.5)--(3,3.5)
(3,-3)--(3.5,-3)
(-3.5,-2)--(-3,-2)
(3,-1)--(3.5,-1)
(-3.5,0)--(-3,0)
(3,1)--(3.5,1)
(-3.5,2)--(-3,2)
(3,3)--(3.5,3);
\draw[]
(-3,3)--(-2,3)--(-2,2) (-1,2)--(-1,3)--(0,3)
(0,2)--(1,2)--(1,3) (2,3)--(2,2)--(3,2)(1,2)--(1,1)(2,2)--(2,1)
(-3,2)--(-3,0)(-3,-1)--(-3,-3)(-3,1)--(-2,1)(-2,-1)--(-2,-3)
(-2,-2)--(-1,-2)(-1,1)--(-1,-1)(-1,0)--(-2,0)(0,0)--(1,0)(0,1)--(0,-1)
(-1,-3)--(0,-3)--(0,-2) (1,-2)--(1,-3)--(2,-3)(1,-1)--(2,-1)(2,0)--(2,-2)
(3,0)--(3,-2);
\draw[fill=white]
(-3,-3) circle [radius=3pt]
(-3,-1) circle [radius=3pt]
(-3,0) circle [radius=3pt]
(-3,2) circle [radius=3pt]
(-3,3) circle [radius=3pt]
(-2,-3) circle [radius=3pt]
(-2,-1) circle [radius=3pt]
(-2,0) circle [radius=3pt]
(-2,1) circle [radius=3pt]
(-2,2) circle [radius=3pt]
(-1,-3) circle [radius=3pt]
(-1,-2) circle [radius=3pt]
(-1,-1) circle [radius=3pt]
(-1,1) circle [radius=3pt]
(-1,2) circle [radius=3pt]
(0,-2) circle [radius=3pt]
(0,-1) circle [radius=3pt]
(0,2) circle [radius=3pt]
(0,3) circle [radius=3pt]
(1,-2) circle [radius=3pt]
(1,-1) circle [radius=3pt]
(1,0) circle [radius=3pt]
(1,1) circle [radius=3pt]
(1,3) circle [radius=3pt]
(2,-3) circle [radius=3pt]
(2,-2) circle [radius=3pt]
(2,0) circle [radius=3pt]
(2,1) circle [radius=3pt]
(2,3) circle [radius=3pt]
(3,-3) circle [radius=3pt]
(3,-2) circle [radius=3pt]
(3,0) circle [radius=3pt]
(3,1) circle [radius=3pt]
(3,2) circle [radius=3pt]
(3,3) circle [radius=3pt]
(0,0) circle [radius=3pt]
(-1,0) circle [radius=3pt]
(2,2) circle [radius=3pt]
(1,2) circle [radius=3pt]
(2,-1) circle [radius=3pt]
(3,-1) circle [radius=3pt]
(0,-3) circle [radius=3pt]
(1,-3) circle [radius=3pt]
(-2,-2) circle [radius=3pt]
(-3,-2) circle [radius=3pt]
(-3,1) circle [radius=3pt]
(-2,3) circle [radius=3pt]
(-1,3) circle [radius=3pt]
(0,1) circle [radius=3pt];
\draw[fill=black]
(1,0)circle [radius=3pt]
(3,2)circle [radius=3pt]
(0,3)circle [radius=3pt]
(-2,1)circle [radius=3pt]
(-1,-2)circle [radius=3pt]
(2,-3)circle [radius=3pt]
(2,0)circle [radius=3pt]
(1,3)circle [radius=3pt]
(-1,1)circle [radius=3pt]
(-3,-1)circle [radius=3pt]
(0,-2)circle [radius=3pt]
(3,-3)circle [radius=3pt];
\end{tikzpicture}
\end{minipage}
\caption{The graph $T_3$. In both figures, the black vertices represent the dominating set $S$ described in the proof of~\Cref{thm:hexagonalGrid}, while the shadowed zones represent how vertices of $S$ dominate their neighbors, thus creating a partition of the vertices of the graph. 
On the left, the thick arrowed edges represent the movement of the guards along the edges when the red vertex is attacked.}
\label{fig:T3}  
\end{figure}



\begin{restatable}[$m$-Eternal Domination on the infinite Triangular Grid]{theorem}{triangularGrid}
\label{thm:triangularGrid}
There is a strongly optimal dominating set for the infinite triangular grid $T_6$ that is also an eternal dominating set. 
\end{restatable}
\begin{proof}
 We define vertex set $S$ on the infinite triangular grid $T_6$ as follows: $(0,0)\in S$; moreover, if $(x,y)\in S$ then also $(x+3,y+1)$, $(x-1,y+2)$, $(x-3,y-1)$ and $(x+1,y-2)$ belong to $S$. 
    This vertex set can be visualized in~\Cref{fig:T6}. 
    Note that the infinite set of closed neighborhoods $N[S]$ of each $s\in S$ is a partition of the vertices of $T_6$; therefore, $S$ is a dominating set of $T_6$.
    
    To prove this formally, observe that for any guard at $(x,y)$ there is a guard at $(x',y')=(x+3,y+1)$, a guard at $(x'',y'')=(x'+3,y'+1)=(x+6,y+2)$, and a guard at $(x''',y''')=(x''+1,y''-2)=(x+7,y)$; 
    Hence, looking at the rows of $T_6$, for each guard at $(x,y)$, the six vertices to its right do not contain any guard.
    Nevertheless, $(x+1,y)$ is dominated by the guard at $(x,y)$; 
    $(x+2,y)$ and $(x+3, y)$ are dominated by the guard at $(x',y')=(x+3, y+1)$;
    $(x+4,y)$ and $(x+5,y)$ are dominated by the guard at $(x'''',y'''')=(x'+1,y'-2)=(x+4,y-1)$;
    finally, $(x+6,y)$ is dominated by the guard at $(x''',y''')=(x+7,y)$.
    The generality of the reasoning proves that
    therefore, 
    $S$ is a dominating set of $T_6$. 

    $S$ is a configuration of an $m$-eternal dominating set of $T_6$. 
    Indeed,
    consider an attack to a vertex $v\in V\setminus S$. 
    Since $S$ is a dominating set and by construction, there exists a unique vertex $s\in S$ that is a neighbor of $v$. 
    It is not restrictive to consider the case $s=(i^*,j^*)$ and $v=(i^*+1,j^*)$ (the other five cases are analogous). 
    For every $(i,j)\in S$, the guard on $(i,j)$ (included $(i^*,j^*)$) moves to $(i+1,j)$. 
    The new position of the guards is a translation of $S$ by one unit in the same direction and thus still forms a dominating set. 
    Therefore, with this strategy, the guards move along configurations of an $m$-eternal dominating set.
    
    We have that $S$ is strongly optimal because for every $s\in S$, $S$ uses only one guard for the seven vertices of $N[s]$, formally 
    $|N[s]\cap S|=1$, and $\{N[s]~|~s\in S\}$ is a partition of $T_6$.\qed    
\end{proof}

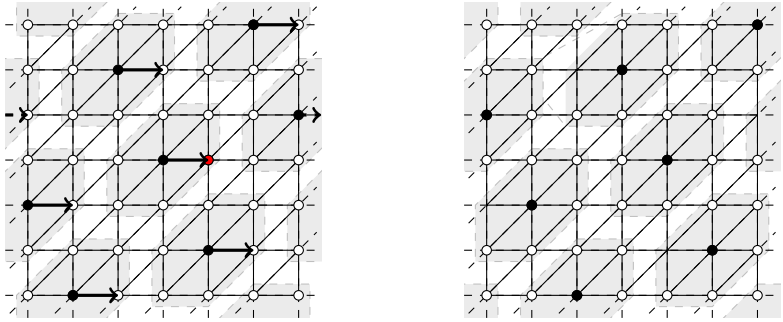
\begin{figure}[ht]
\centering
\hspace*{10mm}
\begin{minipage}{0.5\textwidth}
\begin{tikzpicture}[scale=0.6]

\clip (-3.5,-3.5) rectangle (3.5,3.5);

\fill[lightgray!32](0,1.25) -- (1.25,1.25) -- (1.25,0) -- (0, -1.25) -- (-1.25,-1.25) -- (-1.25, 0) -- cycle;
\draw[dashed, lightgray] (0,1.25) -- (1.25,1.25) -- (1.25,0) -- (0, -1.25) -- (-1.25,-1.25) -- (-1.25, 0) -- cycle;

\fill[lightgray!32](3,2.25) -- (4.25,2.25) -- (4.25,1) -- (3, -.25) -- (1.75,-.25) -- (1.75, 1) -- cycle;
\draw[dashed, lightgray] (3,2.25) -- (4.25,2.25) -- (4.25,1) -- (3, -.25) -- (1.75,-.25) -- (1.75, 1) -- cycle;

\fill[lightgray!32](2,4.25) -- (3.25,4.25) -- (3.25,3) -- (2, 1.75) -- (.75,1.75) -- (.75, 3) -- cycle;
\draw[dashed, lightgray] (2,4.25) -- (3.25,4.25) -- (3.25,3) -- (2, 1.75) -- (.75,1.75) -- (.75, 3) -- cycle;

\fill[lightgray!32](-1,3.25) -- (.25,3.25) -- (.25,2) -- (-1, .75) -- (-2.25,.75) -- (-2.25, 2) -- cycle;
\draw[dashed, lightgray] (-1,3.25) -- (.25,3.25) -- (.25,2) -- (-1, .75) -- (-2.25,.75) -- (-2.25, 2) -- cycle;

\fill[lightgray!32](-3,.25) -- (-1.75,.25) -- (-1.75,-1) -- (-3, -2.25) -- (-4.25,-2.25) -- (-4.25, -1) -- cycle;
\draw[dashed, lightgray] (-3,.25) -- (-1.75,.25) -- (-1.75,-1) -- (-3, -2.25) -- (-4.25,-2.25) -- (-4.25, -1) -- cycle;

\fill[lightgray!32](1,-.75) -- (2.25,-.75) -- (2.25,-2) -- (1, -3.25) -- (-.25,-3.25) -- (-.25, -2) -- cycle;
\draw[dashed, lightgray] (1,-.75) -- (2.25,-.75) -- (2.25,-2) -- (1, -3.25) -- (-.25,-3.25) -- (-.25, -2) -- cycle;

\fill[lightgray!32](-2,-1.75) -- (-.75,-1.75) -- (-.75,-3) -- (-2, -4.25) -- (-3.25,-4.25) -- (-3.25, -3) -- cycle;
\draw[dashed, lightgray] (-2,-1.75) -- (-.75,-1.75) -- (-.75,-3) -- (-2, -4.25) -- (-3.25,-4.25) -- (-3.25, -3) -- cycle;

\fill[lightgray!32](2,-2.75) -- (3.25,-2.75) -- (3.25,-4) -- (2, -5.25) -- (.75,-5.25) -- (.75, -4) -- cycle;
\draw[dashed, lightgray](2,-2.75) -- (3.25,-2.75) -- (3.25,-4) -- (2, -5.25) -- (.75,-5.25) -- (.75, -4) -- cycle;

\fill[lightgray!32](4,.25) -- (5.25,.25) -- (5.25,-1) -- (4, -2.25) -- (2.75,-2.25) -- (2.75, -1) -- cycle;
\draw[dashed, lightgray](4,.25) -- (5.25,.25) -- (5.25,-1) -- (4, -2.25) -- (2.75,-2.25) -- (2.75, -1) -- cycle;

\fill[lightgray!32](-4,2.25) -- (-2.75,2.25) -- (-2.75,1) -- (-4, -.25) -- (-5.25,-.25) -- (-5.25, 1) -- cycle;
\draw[dashed, lightgray] (-4,2.25) -- (-2.75,2.25) -- (-2.75,1) -- (-4, -.25) -- (-5.25,-.25) -- (-5.25, 1) -- cycle;

\fill[lightgray!32](-2,5.25) -- (-.75,5.25) -- (-.75,4) -- (-2, 2.75) -- (-3.25,2.75) -- (-3.25, 4) -- cycle;
\draw[dashed, lightgray] (-2,5.25) -- (-.75,5.25) -- (-.75,4) -- (-2, 2.75) -- (-3.25,2.75) -- (-3.25, 4) -- cycle;

\draw
(-3,-3)--(-3,3)
(-2,-3)--(-2,3)
(-1,-3)--(-1,3)
(0,-3)--(0,3)
(1,-3)--(1,3)
(2,-3)--(2,3)
(3,-3)--(3,3)

(-3,-3)--(3,-3)
(-3,-2)--(3,-2)
(-3,-1)--(3,-1)
(-3,0)--(3,0)
(-3,1)--(3,1)
(-3,2)--(3,2)
(-3,3)--(3,3)

(-3,2)--(-2,3)
(-3,1)--(-1,3)
(-3,0)--(0,3)
(-3,-1)--(1,3)
(-3,-2)--(2,3)
(-3,-3)--(3,3)
(-2,-3)--(3,2)
(-1,-3)--(3,1)
(0,-3)--(3,0)
(1,-3)--(3,-1)
(2,-3)--(3,-2)
;

\draw[dashed] 
(-3,-3.5)--(-3,3.5)
(-2,-3.5)--(-2,3.5)
(-1,-3.5)--(-1,3.5)
(0,-3.5)--(0,3.5)
(1,-3.5)--(1,3.5)
(2,-3.5)--(2,3.5)
(3,-3.5)--(3,3.5)

(-3.5,-3)--(3.5,-3)
(-3.5,-2)--(3.5,-2)
(-3.5,-1)--(3.5,-1)
(-3.5,0)--(3.5,0)
(-3.5,1)--(3.5,1)
(-3.5,2)--(3.5,2)
(-3.5,3)--(3.5,3)

(-3.4,1.6)--(-1.6,3.4)
(-3.4,0.6)--(-0.6,3.4)
(-3.4,-0.4)--(0.4,3.4)
(-3.4,-1.4)--(1.4,3.4)
(-3.4,-2.4)--(2.4,3.4)
(-3.4,-3.4)--(3.4,3.4)

(-2.4,-3.4)--(3.4,2.4)
(-1.4,-3.4)--(3.4,1.4)
(-0.4,-3.4)--(3.4,0.4)
(0.6,-3.4)--(3.4,-0.6)
(1.6,-3.4)--(3.4,-1.6)
;


\draw[fill=white]
(-3,-3) circle [radius=3pt]
(-3,-2) circle [radius=3pt]

(-3,0) circle [radius=3pt]
(-3,1) circle [radius=3pt]
(-3,2) circle [radius=3pt]
(-3,3) circle [radius=3pt]

(-2,-3) circle [radius=3pt]
(-2,-2) circle [radius=3pt]
(-2,-1) circle [radius=3pt]
(-2,0) circle [radius=3pt]
(-2,1) circle [radius=3pt]
(-2,2) circle [radius=3pt]
(-2,3) circle [radius=3pt]

(-1,-3) circle [radius=3pt]
(-1,-2) circle [radius=3pt]
(-1,-1) circle [radius=3pt]
(-1,0) circle [radius=3pt]
(-1,1) circle [radius=3pt]

(-1,3) circle [radius=3pt]

(0,-3) circle [radius=3pt]
(0,-2) circle [radius=3pt]
(0,-1) circle [radius=3pt]

(0,1) circle [radius=3pt]
(0,2) circle [radius=3pt]
(0,3) circle [radius=3pt]

(1,-3) circle [radius=3pt]

(1,-1) circle [radius=3pt]
(1,0) circle [radius=3pt]
(1,1) circle [radius=3pt]
(1,2) circle [radius=3pt]
(1,3) circle [radius=3pt]

(2,-3) circle [radius=3pt]
(2,-2) circle [radius=3pt]
(2,-1) circle [radius=3pt]
(2,0) circle [radius=3pt]
(2,1) circle [radius=3pt]
(2,2) circle [radius=3pt]
(2,3) circle [radius=3pt]

(3,-3) circle [radius=3pt]
(3,-2) circle [radius=3pt]
(3,-1) circle [radius=3pt]
(3,0) circle [radius=3pt]

(3,2) circle [radius=3pt]
(3,3) circle [radius=3pt]
; 

\draw[fill=black]
(0,0) circle [radius=3pt]
(3,1) circle [radius=3pt]
(-1,2) circle [radius=3pt]
(-3,-1) circle [radius=3pt]
(1,-2) circle [radius=3pt]

(-2,-3) circle [radius=3pt]
(2,3) circle [radius=3pt]
;

\draw[fill=red]
(1,0) circle [radius=3pt];

\draw[very thick,->](0,0)--(1,0);
\draw[very thick,->](-1,2)--(0,2);
\draw[very thick,->](-3,-1)--(-2,-1);
\draw[very thick,->](1,-2)--(2,-2);
\draw[very thick,->](-2,-3)--(-1,-3);
\draw[very thick,->](2,3)--(3,3);

\draw[very thick,->,dashed](3,1)--(3.5,1);
\draw[very thick,->,dashed](-3.5,1)--(-3,1);

\end{tikzpicture}
\end{minipage}%
\begin{minipage}{0.5\textwidth}
\begin{tikzpicture}[scale=0.6]

\clip (-3.5,-3.5) rectangle (3.5,3.5);

\fill[lightgray!32](1,1.25) -- (2.25,1.25) -- (2.25,0) -- (1, -1.25) -- (-0.25,-1.25) -- (-0.25, 0) -- cycle;
\draw[dashed, lightgray] (1,1.25) -- (2.25,1.25) -- (2.25,0) -- (1, -1.25) -- (-0.25,-1.25) -- (-0.25, 0) -- cycle;

\fill[lightgray!32](4,2.25) -- (5.25,2.25) -- (5.25,1) -- (4, -.25) -- (2.75,-.25) -- (2.75, 1) -- cycle;
\draw[dashed, lightgray] (4,2.25) -- (5.25,2.25) -- (5.25,1) -- (4, -.25) -- (2.75,-.25) -- (2.75, 1) -- cycle;

\fill[lightgray!32](3,4.25) -- (4.25,4.25) -- (4.25,3) -- (3, 1.75) -- (1.75,1.75) -- (1.75, 3) -- cycle;
\draw[dashed, lightgray] (3,4.25) -- (4.25,4.25) -- (4.25,3) -- (3, 1.75) -- (1.75,1.75) -- (1.75, 3) -- cycle;

\fill[lightgray!32](0,3.25) -- (1.25,3.25) -- (1.25,2) -- (0, .75) -- (-1.25,.75) -- (-1.25, 2) -- cycle;
\draw[dashed, lightgray] (0,3.25) -- (1.25,3.25) -- (1.25,2) -- (0, .75) -- (-1.25,.75) -- (-2.25, 2) -- cycle;

\fill[lightgray!32](-2,.25) -- (-.75,.25) -- (-.75,-1) -- (-2, -2.25) -- (-3.25,-2.25) -- (-3.25, -1) -- cycle;
\draw[dashed, lightgray] (-2,.25) -- (-.75,.25) -- (-.75,-1) -- (-2, -2.25) -- (-3.25,-2.25) -- (-3.25, -1) -- cycle;

\fill[lightgray!32](2,-.75) -- (3.25,-.75) -- (3.25,-2) -- (2, -3.25) -- (.75,-3.25) -- (.75, -2) -- cycle;
\draw[dashed, lightgray] (2,-.75) -- (3.25,-.75) -- (3.25,-2) -- (2, -3.25) -- (.75,-3.25) -- (.75, -2) -- cycle;

\fill[lightgray!32](-1,-1.75) -- (.25,-1.75) -- (.25,-3) -- (-1, -4.25) -- (-2.25,-4.25) -- (-2.25, -3) -- cycle;
\draw[dashed, lightgray] (-1,-1.75) -- (.25,-1.75) -- (.25,-3) -- (-1, -4.25) -- (-2.25,-4.25) -- (-2.25, -3) -- cycle;

\fill[lightgray!32](3,-2.75) -- (4.25,-2.75) -- (4.25,-4) -- (3, -5.25) -- (1.75,-5.25) -- (1.75, -4) -- cycle;
\draw[dashed, lightgray](3,-2.75) -- (4.25,-2.75) -- (4.25,-4) -- (3, -5.25) -- (1.75,-5.25) -- (1.75, -4) -- cycle;

\fill[lightgray!32](5,.25) -- (6.25,.25) -- (6.25,-1) -- (5, -2.25) -- (3.75,-2.25) -- (3.75, -1) -- cycle;
\draw[dashed, lightgray](5,.25) -- (6.25,.25) -- (6.25,-1) -- (5, -2.25) -- (3.75,-2.25) -- (3.75, -1) -- cycle;

\fill[lightgray!32](-3,2.25) -- (-1.75,2.25) -- (-1.75,1) -- (-3, -.25) -- (-4.25,-.25) -- (-4.25, 1) -- cycle;
\draw[dashed, lightgray] (-3,2.25) -- (-1.75,2.25) -- (-1.75,1) -- (-3, -.25) -- (-4.25,-.25) -- (-4.25, 1) -- cycle;

\fill[lightgray!32](-1,5.25) -- (.25,5.25) -- (.25,4) -- (-1, 2.75) -- (-2.25,2.75) -- (-2.25, 4) -- cycle;
\draw[dashed, lightgray] (-1,5.25) -- (.25,5.25) -- (.25,4) -- (-1, 2.75) -- (-2.25,2.75) -- (-2.25, 4) -- cycle;

\fill[lightgray!32](-4,4.25) -- (-2.75,4.25) -- (-2.75,3) -- (-4, 1.75) -- (-5.25,1.75) -- (-5.25, 3) -- cycle;
\draw[dashed, lightgray] (-4,4.25) -- (-2.75,4.25) -- (-2.75,3) -- (-4, 1.75) -- (-5.25,1.75) -- (-5.25, 3) -- cycle;

\fill[lightgray!32](-4,-2.75) -- (-2.75,-2.75) -- (-2.75,-4) -- (-4, -5.25) -- (-5.25,-5.25) -- (-5.25, -4) -- cycle;
\draw[dashed, lightgray](-4,-2.75) -- (-2.75,-2.75) -- (-2.75,-4) -- (-4, -5.25) -- (-5.25,-5.25) -- (-5.25, -4) -- cycle;

\draw
(-3,-3)--(-3,3)
(-2,-3)--(-2,3)
(-1,-3)--(-1,3)
(0,-3)--(0,3)
(1,-3)--(1,3)
(2,-3)--(2,3)
(3,-3)--(3,3)

(-3,-3)--(3,-3)
(-3,-2)--(3,-2)
(-3,-1)--(3,-1)
(-3,0)--(3,0)
(-3,1)--(3,1)
(-3,2)--(3,2)
(-3,3)--(3,3)

(-3,2)--(-2,3)
(-3,1)--(-1,3)
(-3,0)--(0,3)
(-3,-1)--(1,3)
(-3,-2)--(2,3)
(-3,-3)--(3,3)
(-2,-3)--(3,2)
(-1,-3)--(3,1)
(0,-3)--(3,0)
(1,-3)--(3,-1)
(2,-3)--(3,-2)
;

\draw[dashed] 
(-3,-3.5)--(-3,3.5)
(-2,-3.5)--(-2,3.5)
(-1,-3.5)--(-1,3.5)
(0,-3.5)--(0,3.5)
(1,-3.5)--(1,3.5)
(2,-3.5)--(2,3.5)
(3,-3.5)--(3,3.5)

(-3.5,-3)--(3.5,-3)
(-3.5,-2)--(3.5,-2)
(-3.5,-1)--(3.5,-1)
(-3.5,0)--(3.5,0)
(-3.5,1)--(3.5,1)
(-3.5,2)--(3.5,2)
(-3.5,3)--(3.5,3)

(-3.4,1.6)--(-1.6,3.4)
(-3.4,0.6)--(-0.6,3.4)
(-3.4,-0.4)--(0.4,3.4)
(-3.4,-1.4)--(1.4,3.4)
(-3.4,-2.4)--(2.4,3.4)
(-3.4,-3.4)--(3.4,3.4)

(-2.4,-3.4)--(3.4,2.4)
(-1.4,-3.4)--(3.4,1.4)
(-0.4,-3.4)--(3.4,0.4)
(0.6,-3.4)--(3.4,-0.6)
(1.6,-3.4)--(3.4,-1.6)
;


\draw[fill=white]
(-3,-3) circle [radius=3pt]
(-3,-2) circle [radius=3pt]

(-3,0) circle [radius=3pt]
(-3,1) circle [radius=3pt]
(-3,2) circle [radius=3pt]
(-3,3) circle [radius=3pt]

(-2,-3) circle [radius=3pt]
(-2,-2) circle [radius=3pt]
(-2,-1) circle [radius=3pt]
(-2,0) circle [radius=3pt]
(-2,1) circle [radius=3pt]
(-2,2) circle [radius=3pt]
(-2,3) circle [radius=3pt]

(-1,-3) circle [radius=3pt]
(-1,-2) circle [radius=3pt]
(-1,-1) circle [radius=3pt]
(-1,0) circle [radius=3pt]
(-1,1) circle [radius=3pt]

(-1,3) circle [radius=3pt]

(0,-3) circle [radius=3pt]
(0,-2) circle [radius=3pt]
(0,-1) circle [radius=3pt]

(0,1) circle [radius=3pt]
(0,2) circle [radius=3pt]
(0,3) circle [radius=3pt]

(1,-3) circle [radius=3pt]

(1,-1) circle [radius=3pt]
(1,0) circle [radius=3pt]
(1,1) circle [radius=3pt]
(1,2) circle [radius=3pt]
(1,3) circle [radius=3pt]

(2,-3) circle [radius=3pt]
(2,-2) circle [radius=3pt]
(2,-1) circle [radius=3pt]
(2,0) circle [radius=3pt]
(2,1) circle [radius=3pt]
(2,2) circle [radius=3pt]
(2,3) circle [radius=3pt]

(3,-3) circle [radius=3pt]
(3,-2) circle [radius=3pt]
(3,-1) circle [radius=3pt]
(3,0) circle [radius=3pt]

(3,2) circle [radius=3pt]
(3,3) circle [radius=3pt]
(0,0) circle [radius=3pt]
; 

\draw[fill=black]
(1,0) circle [radius=3pt]
(0,2) circle [radius=3pt]
(-2,-1) circle [radius=3pt]
(2,-2) circle [radius=3pt]

(-1,-3) circle [radius=3pt]
(3,3) circle [radius=3pt]

(-3,1)circle [radius=3pt]
;

\end{tikzpicture}
\end{minipage}
\caption{The graph $T_6$. In both figures, the black vertices represent the dominating set $S$ described in the proof of~\Cref{thm:triangularGrid}, while the shadowed zones represent how vertices of $S$ dominate their neighbors, thus creating a partition of the vertices of the graph. On the left, the thick arrowed edges represent the movement of the guards along the edges when the red vertex is attacked.}
\label{fig:T6}  
\end{figure}

\vspace{-10pt}
\section{Conclusions and Open Problems.} We studied the {\sc $m$-Eternal Domination} problem and variants on classes of finite graphs and infinite regular grids. In particular, we showed that {\sc $m$-Eternal Domination} is NP-hard even on bipartite graphs of diameter four. Moreover, the {\sc $m$-Eternal Roman Domination} and {\sc $m$-Eternal Italian Domination} problems are NP-hard even on split graphs. Finally, we showed optimal results for the {\sc Domination} and {\sc $m$-Eternal Domination} problems when considering four types of 
infinite regular grids: square, octagonal, hexagonal, and triangular.

For future work, we propose the following directions. We would like to consider infinite grids with vertex replacement, {\em i.e.}, each vertex in the grid is replaced with a fixed graph $H$ and carrying across edges as complete bipartite graphs or matchings or boundary-based connections. This preserves the overall infinite grid structure, but is also more intricate because, locally, the structure of $H$ would come into play. We have some preliminary observations for simple choices of $H$ and believe this to be a rich direction for future work. It is also interesting to determine the $m$-eternal dominating set for infinite grids with 1,2 or 3 bounded directions, and also extend our results to other domination variants (including the Roman and Italian questions).


%
%
%
\bibliographystyle{splncs04}
\bibliography{refs}

\clearpage
\end{document}